\def\forc{0}
\def\neurips{0}
\def\notes{0}

\documentclass[10pt]{article}
\usepackage[numbers, compress]{natbib}
\usepackage[utf8]{inputenc}

\usepackage{titletoc}

\title{Counting Distinct Elements in the Turnstile Model\\ with Differential Privacy under Continual Observation
\ifnum\forc=1
\footnote{Non-archival track submission. Previously appeared at NeurIPS 2023. 
\href{https://proceedings.neurips.cc/paper_files/paper/2023/hash/0ef1afa0daa888d695dcd5e9513bafa3-Abstract-Conference.html}{[Link to publication]}
}
\else\fi
}

\author{Palak Jain\thanks{Department of Computer Science, Boston University. \texttt{\{palakj,ikalemaj,satchit,sofya,ads22\}@bu.edu}.} 
\and Iden Kalemaj\footnotemark[2]
\and Sofya Raskhodnikova\footnotemark[2]
\and Satchit Sivakumar\footnotemark[2]
\and Adam Smith\footnotemark[2]
 \date{\vspace{-5ex}}
 \date{}}

\usepackage{amsmath,amssymb,amsthm,verbatim,relsize,mathrsfs,soul, url}
\usepackage{titletoc}
\usepackage{hyperref}
\usepackage[page,header]{appendix}
\usepackage{algorithm}
\usepackage[noend]{algpseudocode}
\hypersetup{colorlinks = true,linkcolor = Periwinkle, anchorcolor = red, citecolor = Periwinkle, filecolor = orange,urlcolor = orange}
\usepackage [english]{babel}
\usepackage [autostyle, english = american]{csquotes}
\usepackage{blindtext}
\usepackage[shortlabels]{enumitem}
\usepackage{graphicx}
\usepackage[usenames,dvipsnames]{xcolor}
\usepackage{subcaption}
\usepackage{subfiles}
\usepackage{dsfont}
\usepackage[normalem]{ulem}
\usepackage{bbm}
\usepackage{tikz}
\usetikzlibrary{trees}
\usepackage{cleveref}
\usepackage{xspace}

\ifnum\neurips=1
\usepackage{booktabs}       
\usepackage{nicefrac}       
\usepackage{microtype}      
\else\fi

\usepackage{setspace}

\MakeOuterQuote{"}

\makeatletter
\newcommand\footnoteref[1]{\protected@xdef\@thefnmark{\ref{#1}}\@footnotemark}
\makeatother

\newcommand{\ignore}[1]{}

\newcommand{\R}{\mathbb{R}}
\newcommand{\N}{\mathbb N}

\newcommand{\eps}{\varepsilon}
\newcommand{\poly}{\mathrm{poly}}

\def\polylog{\operatorname{polylog}}

\newcommand{\cN}{\mathcal{N}}
\newcommand{\cA}{\mathcal{A}}
\newcommand{\cB}{\mathcal{B}}
\newcommand{\cY}{\mathcal{Y}}
\DeclareMathSymbol{\mdot}{\mathord}{symbols}{"01}

\newcommand{\abs}[1]{\left\lvert #1 \right\rvert}

\newcommand{\E}{\mathbb{E}}
\newcommand{\cM}{\mathcal{M}}

\ifnum\neurips=0
\usepackage{natbib}
\usepackage{framed}
\usepackage{fullpage}
\fi

\algrenewcommand\algorithmicrequire{\textbf{Input:}}
\algrenewcommand\algorithmicensure{\textbf{Output:}}

\newtheorem{theorem}{Theorem}[section]

\newtheorem{definition}[theorem]{Definition}
\newtheorem{lemma}[theorem]{Lemma}
\newtheorem{corollary}[theorem]{Corollary}
\newtheorem{remark}[theorem]{Remark}

\newtheorem{claim}[theorem]{Claim}

\newcommand{\Sec}[1]{\hyperref[sec:#1]{Section~\ref*{sec:#1}}} 
\newcommand{\Eqn}[1]{\hyperref[eq:#1]{(\ref*{eq:#1})}} 
\newcommand{\Fig}[1]{\hyperref[fig:#1]{Fig.\,\ref*{fig:#1}}} 
\newcommand{\Tab}[1]{\hyperref[tab:#1]{Tab.\,\ref*{tab:#1}}} 
\newcommand{\Thm}[1]{\hyperref[thm:#1]{Theorem\,\ref*{thm:#1}}} 
\newcommand{\Fact}[1]{\hyperref[fact:#1]{Fact\,\ref*{fact:#1}}} 
\newcommand{\Lem}[1]{\hyperref[lem:#1]{Lemma\,\ref*{lem:#1}}} 
\newcommand{\Lems}[2]{\hyperref[lem:#1]{Lemmas\,\ref*{lem:#1}} and~\hyperref[lem:#2]{\ref*{lem:#2}}} 
\newcommand{\Prop}[1]{\hyperref[prop:#1]{Prop.\,\ref*{prop:#1}}} 
\newcommand{\Cor}[1]{\hyperref[cor:#1]{Corollary~\ref*{cor:#1}}} 
\newcommand{\Conj}[1]{\hyperref[conj:#1]{Conjecture~\ref*{conj:#1}}} 
\newcommand{\Def}[1]{\hyperref[def:#1]{Definition~\ref*{def:#1}}} 
\newcommand{\Alg}[1]{\hyperref[alg:#1]{Algorithm~\ref*{alg:#1}}} 
\newcommand{\Proc}[1]{\hyperref[proc:#1]{Procedure~\ref*{proc:#1}}} 
\newcommand{\Step}[1]{\hyperref[step:#1]{Step~\ref*{step:#1}}} 
\newcommand{\Steps}[2]{\hyperref[step:#1]{Steps~\ref*{step:#1}} and~\hyperref[step:#2]{\ref*{step:#2}}} 
\newcommand{\Stepss}[3]{\hyperref[step:#1]{Steps~\ref*{step:#1}},~\hyperref[step:#2]{\ref*{step:#2}}, and~\hyperref[step:#3]{\ref*{step:#3}}} 
\newcommand{\Ex}[1]{\hyperref[ex:#1]{Ex.~\ref*{ex:#1}}} 
\newcommand{\Clm}[1]{\hyperref[clm:#1]{Claim~\ref*{clm:#1}}} 
\newcommand{\Inv}[1]{\hyperref[inv:#1]{Invariant~\ref*{inv:#1}}} 
\newcommand{\Rem}[1]{\hyperref[rem:#1]{Remark~\ref*{rem:#1}}} 
\newcommand{\Obs}[1]{\hyperref[obs:#1]{Observation~\ref*{obs:#1}}} 
\newcommand{\Definition}[1]{\hyperref[definition:#1]{Definition~\ref*{definition:#1}}} 

\ifnum\notes=1
    \definecolor{mygreen}{RGB}{34,164,31}
    \newcommand{\ssnote}[1]{{\color{violet}\footnote{{\color{violet} {\bf SS:} #1}}}}
    \newcommand{\srnote}[1]{{\color{blue}\footnote{{\color{blue} {\bf SR:} #1}}}}
    \newcommand{\pjnote}[1]{{\color{mygreen}\footnote{{\color{mygreen} {\bf PJ:} #1}}}}
    \newcommand{\inote}[1]{{\color{brown}\footnote{{\color{brown} {\bf IK:} #1}}}}
    \newcommand{\asnote}[1]{{\color{RedOrange}\footnote{{\color{RedOrange} {\bf AS:} #1}}}}

    \newcommand{\sstext}[1]{{\color{violet} #1}}
    \usepackage{showlabels}
    
\else 
    \newcommand{\ssnote}[1]{}
    \newcommand{\srnote}[1]{}
    \newcommand{\pjnote}[1]{}
    \newcommand{\inote}[1]{}
    \newcommand{\asnote}[1]{}

    \newcommand{\sstext}[1]{{#1}}

    \renewcommand{\sout}[1]{}
\fi

\newcommand{\CR}{continual release model} 

\renewcommand{\epsilon}{\varepsilon} 
\newcommand{\X}{\mathcal{X}} 
\newcommand{\BigO}[1]{\ensuremath{\operatorname{O}\left(#1\right)}} 
\newcommand{\dset}{y} 
\newcommand{\dstream}{x} 
\renewcommand{\vec}[1]{} 

\newcommand{\wmax}{w_{\max}}

\newcommand{\mech}{\ensuremath{\mathcal{M}}\xspace} 
\newcommand{\alg}{\ensuremath{\mathcal{A}}\xspace} 
\newcommand{\marginals}[1]{\ensuremath{\mathrm{\sf Marginals}_{#1}}} 

\newcommand{\err}{\text{\sf ERR}} 

\newcommand{\msf}{\mathsf}
\newcommand{\countdistinct}{\msf{CountDistinct}}
\newcommand{\innerproducts}{\msf{InnerProducts}}

\newcommand{\existence}[1]{f_{#1}}
\newcommand{\boundedexist}[1]{\tilde{f_{{#1}}}}

\newcommand{\zo}{\{0,1\}} 

\newcommand{\paren}[1]{{\left ( {#1} \right)}}

\newcommand{\flip}{\mathsf{flip}}
\newcommand{\uni}{\mathcal{U}}
\newcommand{\univ}{\uni}
\newcommand{\out}{s}
\newcommand{\abov}{\mathsf{Above}}
\newcommand{\below}{\mathsf{Below}}
\newcommand{\ele}{u}

\begin{document}
\maketitle
\begin{abstract}

Privacy is a central challenge for systems that learn from sensitive data sets, especially when a system's outputs must be continuously updated to reflect changing data.
We consider the achievable error for differentially private continual release of a basic  statistic---the number of distinct items---in a stream where items may be both inserted and deleted (the \textit{turnstile} model). 
With only insertions, existing algorithms have additive error just polylogarithmic in the length of the stream $T$. 
We uncover a much richer landscape in the turnstile model,
even without considering memory restrictions. 
We show that every differentially private mechanism that handles insertions and deletions has \textit{worst-case} additive error  at least $T^{1/4}$ even under  a relatively weak, \textit{event-level} privacy definition. 
Then, we identify a parameter of the input stream, its \textit{maximum flippancy}, that is low for natural data streams and for which we give tight parameterized error guarantees. Specifically, the maximum flippancy is the largest number of times that the contribution of a single 
item to the distinct elements count changes over the course of the stream. We 
present an \textit{item-level} differentially private mechanism that, for all turnstile streams with  maximum flippancy  $w$, continually outputs the number of distinct elements with an $O(\sqrt{w} \cdot \mathsf{poly}\log T)$ additive error,  without requiring prior knowledge of $w$. 
We prove that this is the best achievable error bound  that depends only on $w$, for a large range of values of $w$.
When $w$ is small, 
the error of our mechanism is similar to the polylogarithmic in $T$ error in the insertion-only setting, bypassing the hardness in the turnstile model. 
\end{abstract}

\section{Introduction}\label{sec:intro} 


Machine learning algorithms are frequently run on sensitive data. In this context, a central challenge is to protect the privacy of individuals whose information is contained in the training set. Differential privacy \cite{DworkMNS16} provides a rigorous framework for the design and analysis of algorithms that publish aggregate statistics, such as parameters of machine learning models, while preserving privacy. In this work, we focus on the model of differential privacy interchangeably called {\em continual observation} and {\em continual release} that was introduced by ~\citet{DworkNPR10} and~\citet{ChanSS11} to study privacy in settings when both the data and the published statics are constantly updated. 
One of the most fundamental statistics about a data stream is the number of distinct elements it contains (see, e.g., the book by Leskovec et al.~\cite{LeskovecRU14}). The problem of counting distinct elements has been widely studied,
starting with the work of Flajolet and Martin~\cite{FlajoletM85}, and has numerous applications~\cite{Akella2003,EstanVF06,MetwallyAA08, HeuleNH13,Baker2018}, including monitoring the number of logins from distinct accounts to a streaming service, tracking the number of different countries represented by people in a chat room, and tracking the number of students signed up for at least one club at a university. Algorithms for this problem are also used as basic building blocks in more complicated data analyses.

We investigate the problem of privately
counting the number of distinct elements under continual observation in the turnstile model, which allows both
element insertions and deletions. 
In the continual release model, a data collector receives a sensitive dataset as a stream of inputs and produces, after receiving each input, an output that is accurate for all inputs received so far.
The input stream is denoted $\dstream$ and its length (also called the {\em time horizon}) is denoted $T$. The elements come from a universe $\uni$. Each entry in the stream is 
an {\em insertion} (denoted by $+\ele$) or a {\em deletion}  (denoted by $-\ele$) of some element $\ele \in \uni$ or, alternatively, a {\em no-op} (denoted by $\bot$) that represents that no update occurred in the current time step.
More formally, for a universe $\uni$, let $\uni_{\pm}$ denote the set  $\{+,-\}\times \uni
\cup \{\bot\}$ of possible stream entries. The shorthand $+\ele$ and $-\ele$ is used for the pairs $(+,\ele)$ and $(-,\ele)$.  
Given a vector $\dstream$ of length $T$ and an integer $t\in[T]$, the vector $\dstream[1:t]$ denotes the prefix of $\dstream$ consisting of the first $t$ entries of $\dstream$. 

Next, we define the function $\countdistinct$ in the (turnstile) continual release model.
 
\begin{definition}[Existence vector, $\countdistinct$]\label{def:existence-vector} Fix a universe $\uni$ and a time horizon $T \in \N$. For an element $u\in\uni$ and a stream $\dstream \in \uni_{\pm}^T$, the {\em existence vector} 
$\existence{\ele}(\dstream) \in \{0,1\}^T$ is an indicator vector that tracks the existence of element $u$ in $\dstream$: specifically, for each $t\in[T]$, the value
$\existence{\ele}(x)[t]$ is $1$ if and only if 
there are strictly more insertions than deletions
of element $u$ in $\dstream[1:t].$
The function $\countdistinct{}:\uni_\pm^T \to \N^T$ returns a vector of the same length as its input, where 
$\countdistinct(x)[t] = \sum_{\ele \in \uni}\existence{\ele}(x)[t]$ for all $t\in [T]$.  
\end{definition}
The focus of our investigation is the best achievable error
in the continual release model for a given time horizon $T$ and privacy parameters. 
We study the worst-case (over all input streams and time steps $t$) additive error of privately approximating the distinct elements counts under continual release. 
\begin{definition}[Error of an answer vector and error of a mechanism for $\countdistinct$]
    Given an answer vector $a\in \R^T$, the error of this vector with respect to the desired function value $f(\dstream)\in \R^T$ computed on dataset $\dstream$ is defined as
    $\err_f(\dstream,a)=\|f(\dstream)-a \|_\infty.$
A mechanism for $\countdistinct$ in the continual release model is $\alpha$-accurate if it outputs a vector of answers $a$ with error  $\err_{\countdistinct}(\dstream,a) \leq \alpha$ with probability at least 0.99.
\end{definition}

Next, we discuss privacy. Originally, differential privacy~\cite{DworkMNS16} was defined in a setting where a data collector outputs the desired information about an entire dataset all at once.  We call this the {\em batch model} to contrast it with
continual release.
In the batch model, two datasets are called {\em neighbors} if they differ in the data of one individual. There are two natural ways to adapt this definition to the continual release model~\cite{DworkNPR10,ChanSS11}, depending on the desired privacy guarantees. 

\begin{definition}[Neighboring streams] \label{def:neigbhors}
Let $x, x' \in \uni_\pm^T$ be two streams of length $T$. Streams $x$ and  $x'$  are {\em event-neighbors} if one can be obtained from the other by replacing a stream entry with $\perp$. Streams $x$ and  $x'$ are {\em item-neighbors} if one can be obtained from the other by replacing a subset of stream entries pertaining to one specific element of $\uni$ with symbols $\perp$.
\end{definition}
Differential privacy can be defined with respect to any notion of neighboring datasets. There are two privacy parameters: $\eps>0$ and $\delta\in [0,1)$. 
An algorithm $\alg$ is $(\eps,\delta)$-{\em differentially private (DP)} if for all pairs of neighboring datasets $\dstream,\dstream'$ and all events $S$ in the output space of $\cA$, 
\begin{align*}
        \Pr[\cA(\dstream) \in S] \leq e^\eps \Pr[\cA(\dstream') \in S] + \delta. 
\end{align*}
The case when $\delta=0$ is referred to as {\em pure} differential privacy, and the general case as {\em approximate} differential privacy.
For event-neighboring (respectively, item-neighboring) streams $\dstream,\dstream' \in \uni_{\pm}^T$, we say that $\cA$ is {\em $(\eps, \delta)$-event-level-DP} (respectively, {\em item-level-DP}). Item-level differential privacy imposes 
a more stringent requirement than event-level, since it guards against  larger changes in the input stream. To contrast with the batch setting, we refer to continual release algorithms as {\em mechanisms}.

In the batch setting, where only $\countdistinct(x)[T]$ is released, there is an $\eps$-DP algorithm for counting distinct elements with expected error $O(1/\eps)$ since the function $\countdistinct(x)[T]$ has sensitivity 1---regardless of 
whether we consider event-level or item-level privacy. 
Privacy is more challenging in the continual release setting, where we aim to release a sequence of estimates, one for each time $t$, and we require that 
the privacy guarantee hold for the entire sequence of outputs. 
Prior work on privately estimating distinct elements in this setting  considered the insertion-only model, exclusively: Bolot et al.~\cite{BolotFMNT13} show that one can get a sequence of estimates, all of  which are within additive error $\mathrm{poly}(\log T)/\eps$. Their result holds for both item-level and event-level privacy (which are essentially equivalent for counting distinct elements with only insertions). Follow-up work generalized their mechanism but, again, considered only insertions~\cite{Ghazi0NM23,EpastoMMMVZ23}.

We uncover a much richer landscape in the turnstile model,
even without considering memory restrictions. 
We show that any differentially private mechanism that 
handles insertions and deletions has \textit{worst-case} additive error  at least $T^{1/4}$ even under \textit{event-level} privacy, the weaker of the two privacy notions. 
To overcome this lower bound, we identify a property of the input stream, its \textit{maximum flippancy}, that is low for natural data streams and for which one can give tight parameterized error guarantees. To define flippancy, recall the 
notion
of the existence vector from \Def{existence-vector}.

\begin{definition}[Flippancy] 
Given a stream $x$ of length $T$ and an element $\ele \in \uni$, the {\em flippancy} of $\ele$ in $x$, denoted by $\flip(\ele, x)$, is the number of pairs of adjacent entries in the existence vector $\existence{\ele}(x)$ with different values. That is, 
$
\flip(\ele, x)
= | \{ j \in [T-1] : \existence{\ele}(x)[j] \neq \existence{\ele}(x)[j + 1]  \} |.
$
The maximum flippancy of a stream $x$, denoted $w_x$, is $\max_{\ele \in \uni}\flip(\ele, x)$.
\end{definition}
In other words, the maximum flippancy is the largest number of times 
the contribution of a single item to the distinct elements count changes
over the course of the stream. 
We design item-level private mechanisms 
whose error scales with the maximum flippancy of the  stream, even though the maximum flippancy is not an input to the mechanism. We show matching lower bounds for item-level privacy that hold in all parameter regimes.  For a large range of the flippancy parameter, 
we also show a matching lower bound for event-level privacy, via a different argument. This leaves a range with an intriguing gap between item-level and event-level bounds.

\subsection{Our results}\label{sec:results}

\begin{table*}[t] 
\begin{center} 
\begin{tabular}{|c||l r|l r|} 
      \hline
Bounds & \multicolumn{2}{c|}{Event-Level Privacy}
 & \multicolumn{2}{c|}{Item-Level Privacy} 
 \\
 \hline
 \hline
 \begin{tabular}{c} \vspace{5pt}\\Upper 
 \\ \end{tabular} 
 & \multicolumn{3}{c}{$ $} & 
 \\[-2.2em]
 & 
 \multicolumn{3}{l}
 {\hspace{9mm} 
$\tilde{O}\left(\mathsf{min}\left(\left(\sqrt{w_x}\log T + \log^3T\right) \cdot \frac{\sqrt{\log 1/\delta}}{\eps}, \frac{\left(T \log 1/\delta\right)^{1/3}}{\eps^{2/3}}, T \right)\right)  $}
 & 
 \small
 {(Thm.~\ref{thm:item-epsdel})}
 \\
 & \multicolumn{3}{c}{$ $} & \\[-0.2em]
\hline
\begin{tabular}{c} Lower 
\\ \end{tabular} & 
$\Omega\left(\mathsf{min}\Big(\frac{\sqrt{w_x}}{\eps},\frac{T^{1/4}}{\eps^{3/4}}, T\Big)\right)$
&
\small
{(Thm.~\ref{thm:insertion+deletion+eps})}
&
$ \tilde{\Omega}\Big(\mathsf{min} \Big( \frac{\sqrt{w_x}}{\eps}, \frac{T^{1/3}}{\eps^{2/3}}, T  \Big)\Big)$
&
{(Thm.~\ref{thm:LB_CD_item})}
\\ 
\hline
\end{tabular}
\end{center}
\caption{Summary of our results: 
bounds on the worst-case additive error
for $\countdistinct$ under event-level and item-level $(\eps,\delta)$-differential privacy, with $\eps \leq 1$ and $\delta = o(\frac{\eps}{T})$. The upper bound depends
on the maximum flippancy $w_x$ of the input $x$, for every $x$.  The lower bounds apply to the worst-case error of an algorithm taken over all inputs with a given maximum flippancy.
}
\label{tab:results}
\end{table*}

Our results are summarized in Table~\ref{tab:results}. 
Our first result is a mechanism for privately approximating $\countdistinct$ for turnstile streams. 
For a stream $x$ of length $T$ with maximum flippancy $w_x$, this mechanism is item-level-DP and has error $\BigO{\min(\sqrt{w_x}\mdot\polylog T, T^{1/3}})$.
Crucially, the mechanism is not given the maximum flippancy up front. 

\begin{theorem}[Upper bound]\label{thm:item-epsdel}
    For all $ \eps,\delta \in (0,1]$ and sufficiently large $T \in \N$, there exists an $(\eps, \delta)$-item-level-DP mechanism for $\countdistinct$ that is  $\alpha$-accurate for all turnstile streams $x$ of length $T$, where
    $$\alpha = \tilde{O}\left(\mathsf{min}\left(\left(\sqrt{w_x}\log T + \log^3T\right) \cdot \frac{\sqrt{\log 1/\delta}}{\eps}, \frac{\left(T \log 1/\delta\right)^{1/3}}{\eps^{2/3}}, T \right)\right),$$
    and $w_x$ is the maximum flippancy of the stream $x$. 
\end{theorem}

Since this mechanism is item-level-DP, it is also event-level-DP with the same privacy parameters. The error it achieves is the best possible in terms of dependence
only on $w_x$ for item-level privacy, and this error is nearly tight for event-level privacy. 
When $w_x$ is small, as is the case for many natural streams, our mechanism has error $O(\polylog T)$, similar to mechanisms for the insertion-only setting.  

\Thm{item-epsdel} can be easily extended to $\eps$ bounded by any constant larger than $1$. We fixed the bound to be $1$ to simplify the presentation. 
Our mechanism has polynomial time and space complexity in the input parameters, although it does not achieve the typically sublinear space guarantees of streaming algorithms. (See ``Bounded Memory'' in \Cref{sec:limitations} for discussion.)

Our lower bounds on the accuracy for $\countdistinct$ for turnstile streams 
are parametrized by a flippancy bound $w$, and apply for streams with maximum flippancy $w_x \leq w$. For event-level DP, our lower bound shows that for all mechanisms with error guarantees expressed solely in terms of the maximum flippancy $w_x$, time horizon $T$, and privacy parameter $\eps$, our $\countdistinct$ mechanism is asymptotically optimal for a large range of values of $w_x$, namely, for all $w _x\leq T^{1/2}$ and $w_x \geq T^{2/3}$. The best achievable error for $w_x \in (T^{1/2}, T^{2/3})$ for event-level differential privacy remains an open problem. 

\begin{theorem}[Event-level lower bound]\label{thm:insertion+deletion+eps}
For all $\eps, \delta \in (0,1]$, sufficiently large $w,T \in \N$ such that $w \leq T$, and all 
$(\eps, \delta)$-event-level-DP 
mechanisms that are $\alpha$-accurate for $\mathsf{CountDistinct}$ on turnstile streams of length $T$ with maximum flippancy at most $w$, if $\delta = o(\frac{\eps}{T})$, then
$$\alpha = \Omega\left(\mathsf{min}\left(\frac{\sqrt{w}}{\eps},\frac{T^{1/4}}{\eps^{3/4}}, T\right)\right).$$
\end{theorem}

In particular, for every $w$, there exists a family of streams satisfying $w_x=\Theta(w)$, for which every mechanism has worst-case error $\Omega\left(\mathsf{min}\left(\frac{\sqrt{w_x}}{\eps},\frac{T^{1/4}}{\eps^{3/4}}, T\right)\right)$. This is reflected in \Cref{tab:results}.


For item-level-DP, our lower bound on the error matches our upper bound for all regimes of $w_x$ up to polylogarithmic factors. 

\begin{theorem}[Item-level lower bound]\label{thm:LB_CD_item}
Let $\eps \in (0,1]$, $\delta \in (0,1]$, and sufficiently large $w,T \in \N$ such that $w \leq T$. For all $(\eps, \delta)$-item-level-DP 
mechanisms that are $\alpha$-accurate for $\mathsf{CountDistinct}$ on turnstile streams of length $T$ with maximum flippancy at most $w$: \\
\hspace*{3mm} {\bf 1} If $\delta = o(\eps/T)$, then $\alpha = \tilde{\Omega}\Big(\mathsf{min} \Big( \frac{\sqrt{w}}{\eps}, \frac{T^{1/3}}{\eps^{2/3}}, T)  \Big)\Big)$.\\
\hspace*{3mm}  {\bf 2} If $\delta = 0$, then $\alpha = \Omega\Big(\mathsf{min} \Big( \frac{w}{\eps}, \sqrt{\frac{T}{\eps}}, T  \Big)\Big) $.
\end{theorem}

Like for \Thm{insertion+deletion+eps}, we can replace $w$ with $w_x$ in the lower bounds stated in \Thm{LB_CD_item} in the following sense: for every $w$, there exists a family of streams satisfying $w_x=\Theta(w)$, for which every mechanism has worst-case error as stated, but with $w$ replaced with $w_x$. (This applies to both Item 1 and Item 2.) For Item~1, this is reflected in \Cref{tab:results}.

\paragraph{Variants of the model.}
All our lower bounds also hold in the \textit{strict turnstile model}, where element counts never go below $0$. They also apply to \emph{offline} mechanisms that receive the entire input stream before producing output; they do not rely on the mechanism's uncertainty about what comes later in the stream.
Furthermore, our item-level lower bounds hold even in the model where each element can be inserted only when it is absent and deleted only when it is present. We call this {\em the likes model}, since it captures the behavior of
the ``like'' counts on social media websites. Our event-level lower bound does not apply to 
the likes
model, which admits an event-level-DP mechanism for $\countdistinct$ with error polylogarithic in $T$, as explained in \Cref{sec:related} in the discussion of the work of \citet{HenzingerSS23}.
In Appendix~\ref{sec:other_lower_bounds}, we provide formal statements and proof outlines for the aforementioned lower bounds for the variants of our model.

\subsection{Our techniques}\label{sec:techniques} 

\paragraph{Upper bound techniques: tracking the maximum flippancy.} 
Before describing our algorithmic ideas, we explain the main obstacle to using the techniques previously developed for insertion-only streams~\cite{BolotFMNT13,EpastoMMMVZ23}  in the turnstile setting.
\citet{BolotFMNT13} and \citet{EpastoMMMVZ23} used a reduction 
from $\countdistinct$ to the summation problem.  
A mechanism for the summation problem outputs, at every time step $t \in [T]$,  the sum of the first $t$ elements of the stream.  Dwork et al.~\cite{DworkMNS16} and Chan et al.~\cite{ChanSS11} designed the binary-tree mechanism to obtain a $O(\polylog T)$-accurate mechanism for summation.
Given an input stream $x$ of length $T$ (to the $\countdistinct$ problem), define a corresponding summation stream $s_x \in \{-1,0,1\}^T$. At time step $t \in [T]$, the entry $s_x[t]$ equals the difference in the count of distinct elements between time steps $t-1$ and $t$, i.e., $s_x[t] = \countdistinct(x)[t] - \countdistinct(x)[t-1]$. Then $\countdistinct(x)[t]$ is precisely the sum of the first $t$ elements of $s_x$.  In the insertion-only model, changing one entry of $x$ changes at most $2$ entries of $s_x$, and thus, by group privacy, the binary-tree mechanism has $O(\polylog T)$ additive error for $\countdistinct$. For turnstile streams, even under the weaker notion of event-level privacy, a change in the stream $x$ can cause $\Omega(T)$ changes to $s_x$. To see this, consider the stream consisting of alternating insertions ($+\ele$) and deletions ($-\ele$) of a single element $\ele \in \uni$, and its event-neighboring stream where the first occurrence of $+\ele$ is replaced with $\perp$. This example illustrates that one of the difficulties of the $\countdistinct$ problem for turnstile streams lies with items that switch from being present to absent multiple times over the course of the stream, that is, have high flippancy. 
We present a private mechanism that outputs estimates of the count of distinct elements in a turnstile stream with optimal accuracy in terms of maximum flippancy.

Our first key idea allows us to obtain a mechanism, \Alg{item-upper}, that is given as input a flippancy upper bound $w$. For a stream $x$ whose maximum flippancy is bounded by $w$, changing to an item-neighbor of $x$ (whose maximum flippancy is also bounded by $w$)
causes at most $2w$ changes to the corresponding summation stream $s_x$. 
This observation, combined with a group privacy argument, gives a mechanism with error $O(w \cdot \polylog\ T)$ directly from the accuracy guarantee of the binary-tree mechanism for summation. 
Previous works in the insertion-only model \cite{BolotFMNT13,EpastoMMMVZ23} used precisely this approach for the special case
$w=1$. 
To obtain the better $\sqrt{w}$ dependence on $w$ in our upper bound, we ``open up'' the analysis of the binary-tree mechanism. By examining the information stored in each node of the binary tree for the summation stream, we show that changing the occurrences of one item in a stream $x$ with maximum flippancy at most $w$ can change the values of at most $w$ nodes in each \emph{level} of the binary tree. The $\sqrt{w}$ dependence in the error then follows from the privacy guarantees of the Gaussian mechanism (used inside the binary-tree mechanism) for approximate differential privacy. 
This type of noise reduction makes crucial use of the binary tree approach: there are optimized noise addition schemes for prefix sums that improve quantitatively over the binary-tree mechanism (see, e.g., \cite{DenisovMRST22, HenzingerUU23}), but it is unclear if they allow 
the same noise reduction. While our mechanism is only accurate for streams with maximum flippancy at most $w$, it is private even for streams that violate this condition. To achieve this, our mechanism ignores stream elements after their flippancy exceeds $w$. 

The second key idea allows our algorithms to adapt automatically to the maximum flippancy $w_x$ of the input, without the need for
an a-priori bound $w$.
We design a private mechanism, \Alg{item-w}, 
that approximately keeps track of the maximum flippancy of the prefix of the stream seen so far and invokes our first mechanism (\Alg{item-upper}) with the current estimated maximum flippancy $\hat w$ as an input. 
Our main innovation lies in the careful application of the sparse vector algorithm~\cite{DworkNRRV09} to track the maximum flippancy of the stream.  We cannot do this directly, since the sparse vector algorithm achieves good utility only for queries of low sensitivity, 
and maximum flippancy has global sensitivity $\Omega(T)$ under both event-level and item-level changes.

Instead, we track a low sensitivity proxy that indirectly monitors the maximum flippancy $w_x$: given the current estimate $\hat w$ of the flippancy, we use the sparse vector algorithm to continuously query
\emph{the number of items in the stream with flippancy above $\hat w$}.
This query has sensitivity one for item-level neighbors, as desired,
but it is not a priori  clear how to use it to upper bound the maximum flippancy of the stream. This is remedied by observing that \Alg{item-upper}, invoked with a flippancy bound $\hat w$, has the same error (and privacy) guarantee even if at most $\sqrt{\hat w}$ items in the stream have flippancy higher than $\hat w$. That is, an exact upper bound on the maximum flippancy is not needed to design an accurate mechanism. Items that violate the flippancy bound are safely ignored by \Alg{item-upper} and do not contribute to the distinct elements count. 
%

When the number of high-flippancy items gets large, we adjust the estimate $\hat w$ and invoke a new instantiation of \Cref{alg:item-upper}. By doubling $\hat w$ each time this happens, we ensure that it remains at most twice the actual maximum flippancy $w_x$, and that we need only invoke $\log T$ different copies of \Cref{alg:item-upper} and the sparse vector algorithm\footnote{All $\log$ expressions in this article are base $2$.}. 
With these ideas, we obtain an item-level private mechanism that, for all streams $x$, has error that scales with $\sqrt{w_x}$.

\paragraph{Lower bound techniques. } Our lower bounds use the embedding technique introduced by~\citet{JainRSS23} to obtain strong separations between the batch and continual release models of differential privacy. 
The approach of Jain~et~al.~embeds multiple separate instances of an appropriately chosen base problem \textit{on the same sensitive dataset} in the batch model into a single instance of a continual release problem. 
Then, the continual release mechanism can be used to solve multiple instances of the base problem in the batch model. The hardness results in the continual release model follow from lower bounds for the batch model. 

A key idea in our event-level lower bound is a connection between the inner product of two binary vectors and the count of distinct elements in the union of those indices where the vector bits equal $1$. Estimates of distinct elements counts can thus be used to estimate inner products on a sensitive dataset of binary bits. Lower bounds on the accuracy of private algorithms for estimating inner product queries have been previously established in the batch model through the reconstruction attack of Dinur and Nissim~\cite{DinurN03}. This connection was used by Mir~et~al.~\cite{MirMNW11} to provide lower bounds for pan-private algorithms for counting distinct elements. However, continual release and pan-privacy are orthogonal notions, and their results don't imply any lower bounds in our setting.
We crucially use deletions to embed multiple instances of inner product queries into a stream: once a query is embedded and the desired estimate is received, the elements inserted to answer that query can be entirely deleted from the stream to obtain a ``clean slate'' for the next query. We obtain a lower bound of $T^{1/4}$ on the error of event-level private mechanisms for $\countdistinct$ in turnstile streams. 

We obtain our stronger item-level lower bounds (for pure and approximate differential privacy) by embedding multiple instances of a $1$-way marginal query. 
 We then apply lower bounds of~\citet{HardtT10} and~\citet{BunUV18} for releasing all 1-way marginals in the batch model in conjunction with our reduction. 
 The 1-way marginals of a dataset $y \in \{0,1\}^{n \times d}$, consisting of $n$ records and $d$ attributes, are the averages of all $d$ attributes of $y$. Deletions in the stream are once again crucially used to embed a marginal query for one attribute and then clean the slate for the next attribute. Changing one record/row in the dataset $y$ translates to $d$ changes of an item in the constructed stream, and thus this reduction is particularly tailored to item-level lower bounds.

\subsection{Related work}\label{sec:related}

The study of differential privacy under continual release was initiated by two concurrent works \cite{DworkNPR10, ChanSS11}. They proposed the binary-tree mechanism for computing sums of binary bits. 
The versatility of this mechanism is demonstrated by numerous applications it has found in the continual release setting and elsewhere.
Under continual release, it has been extended to work for sums of real values \cite{PerrierAK19}, weighted sums \cite{BolotFMNT13}, graph statistics \cite{SongLMVC18, FichtenHO21}, and most relevantly, counting distinct elements \cite{BolotFMNT13, EpastoMMMVZ23, Ghazi0NM23}. 
It has also been employed for private online learning \cite{jainkt12, SmithT13, AgarwalS17} and for answering range queries \cite{DworkNPR10, DworkNRR15, EdmondsNU20}.

    Prior to our work, the $\countdistinct$ problem with continual release was studied exclusively in the insertions-only model. \citet{BolotFMNT13} were the first to study this problem and showed a $O(\log^{1.5}T)$-accurate item-level-DP mechanism. \citet{Ghazi0NM23} considered the more challenging sliding-window model and showed nearly-matching upper and lower bounds for this setting, parameterized by the window size, for item-level and event-level differential privacy. \citet{EpastoMMMVZ23} studied the more general $\ell_p$-frequency estimation problem with a focus on space efficiency. For distinct elements, i.e., $p=0$, their mechanism provides an estimate with $1+\eta$ multiplicative error and $O(\log^2T)$ additive error, using space $\poly (\log T/\eta)$. They also extended their results to the sliding-window model. Two of the works~\cite{BolotFMNT13, EpastoMMMVZ23} reduced the $\countdistinct$ problem to the bit summation primitive, which allowed them to use the binary-tree mechanism. Since the streams are restricted to be insertion-only, the bit summation primitives they considered have low constant sensitivity. The same primitives have sensitivity $\Omega(T)$ for turnstile streams, and thus this approach cannot be directly extended to our setting. \citet{Ghazi0NM23} observed that for fixed and sliding windows, the distinct elements problem can be reduced to range queries. For the special case when the window is the entire stream, their reduction is to the summation problem. 

In concurrent work, \citet{HenzingerSS23} studied $\countdistinct$ with insertions and deletions in the likes model (see ``Variants of the model'' in \Cref{sec:results}).
 Our model is more general, since it allows for
 multiple consecutive insertions and deletions of the same item. 
Our upper bound and our item-level privacy lower bound can be extended to the likes model. In contrast, our event-level private lower bound provably does not apply to that model: in the likes model, for event-level privacy, there is a simple reduction to the bit summation problem in the continual release model such that the resulting algorithm incurs only a polylogarithmic in $T$ error, whereas we show that in our model, any event-level private algorithm incurs a polynomial in $T$ error. 
 
\citet{HenzingerSS23} showed error bounds for item-level privacy in the likes model that are parameterized by the total number of updates $K$ in the stream. 
The parameter $K$ is related to our concept of flippancy: in the likes model, $K$ equals the  sum of all items' flippancies  and, in general, is at least that sum.
\citet{HenzingerSS23} give an $(\eps, 0)$-DP algorithm with error $\tilde{O}(\sqrt{K} \log T)$ and show a nearly matching lower bound on the error for $(\eps, 0)$-DP algorithms using a packing argument. This lower bound applies to our model as well. 
It is incomparable to our lower bounds, since it scales differently and depends on a different parameter. In our model, their algorithm can be analyzed to give error bounds in terms of the sum $K'$ of the flippancies of the items and incurs error $\tilde{O}(\sqrt{K'} \log T)$; however, it is unclear if their algorithm can be analyzed in our model to give bounds in terms of the (potentially smaller) maximum flippancy.

    Another line of work investigated private sketches for distinct elements, motivated by the popularity of sketching algorithms for the streaming setting.  Mergeable sketches for counting distinct elements have received particular attention~\cite{StanojevicNY2017, ChoiDKY2020, PaghS21, HehirTC2023}, since they allow multiple parties to estimate the joint count of distinct elements by merging their private sketches. While these sketches can be combined with the binary-tree mechanism to obtain private mechanisms for $\countdistinct$, the utility deteriorates when many $(\log T)$ sketches are merged. In fact, \citet{DesfontainesLB2018} showed that achieving both privacy and high accuracy is impossible when many sketches for counting distinct elements are merged. 
    Other private sketches have been studied \cite{SmithST2020, DickensTT2022, WangPS22} for the streaming 
    batch setting  (without continual release).
    The distinct elements problem has also been studied in a distributed setting
    \cite{ChenG0M21,GhaziKKMPSWW22} and under pan-privacy \cite{MirMNW11}. In particular, our lower bound for event-level privacy uses ideas from the lower bound of \citet{MirMNW11}, as described in \Sec{techniques}

    The $\countdistinct$ problem has been extensively studied in the non-private streaming setting, where the goal is to achieve low space complexity \cite{FlajoletM85, Alon1996,Cohen1997,Gibbons2001,GibbonsT01,BarYossef2002,Bar-YossefJKST02,Durand2003,IndykW03,Woodruff04,EstanVF06,Beyer2007,Flajolet2007,Brody2009, KaneNW10}.  \citet{BlockiGMZ23} showed a black-box transformation for every streaming algorithm with tunable accuracy guarantees into a DP algorithm with similar accuracy, for low sensitivity functions. Their transformation does not obviously extend to the continual release setting. Moreover $\countdistinct$ has high sensitivity for turnstile streams.

    The first lower bound in the continual release model of differential privacy was an $\Omega(\log T)$ bound on the accuracy of mechanisms for bit summation, shown by \citet{DworkNPR10}. \citet{JainRSS23} gave the first polynomial separation in terms of error between the continual release model and the batch model under differential privacy. Our lower bounds also show such a separation. The lower bounds of \citet{JainRSS23} were for the problems of outputting the value and index of the attribute with the highest sum, amongst $d$ attributes of a dataset.  Our lower bounds are inspired by their sequential embedding technique to reduce multiple instances of a batch problem to a problem in the continual release model. Similar to them, we also reduce from the 1-way marginals problem to obtain our item-level lower bound. However, our event-level lower bound involves reducing from a different problem, and our reductions use the specific structure of $\countdistinct$ for turnstile streams. 

\subsection{Broader impact, limitations, and open questions}
\label{sec:limitations}

We study the achievable error of DP mechanisms for $\countdistinct$ under continual observation 
in streams with insertions and deletions.
We show that it is characterized by the \textit{maximum flippancy} of the stream. Our work is motivated by societal concerns, but focused on fundamental theoretical limits. It contributes to the broader agenda of obtaining privacy-preserving algorithms for data analysis. We discuss natural directions for future research and some limitations of our work.

   \textbf{Tight bounds:} We found the best achievable error in some settings, but our upper and lower bounds do not match in some parameter regimes. What is the right error bound for event-level privacy for streams $x$ with maximum flippancy $w_x$ between $\sqrt{T}$ and $T^{2/3}$? 
   Our results yield a lower bound of $T^{1/4}$
   and an upper bound of roughly $\sqrt{w_x}$.

   \textbf{Bounded memory: } We did not consider
    any memory restrictions. 
    Prior to our work, no other work addressed $\countdistinct$ with deletions under continual release—with or without space constraints. We consider only the privacy constraint since
    it is more fundamental—it cannot be avoided by buying more memory—and the best algorithms with unbounded memory provide a benchmark by which to evaluate space-constrained approaches.

    Space complexity is certainly a natural topic for future work. While
    it is not clear how to apply the sketching techniques of \citet{EpastoMMMVZ23} to the turnstile setting, it would be interesting to come up with accurate, private, and low-memory mechanisms for counting distinct elements in turnstile streams. Such algorithms would necessarily mix multiplicative and additive error guarantees (due to space and privacy constraints, respectively).

\section{Additional background on differential privacy}\label{sec:dp}

In this section, we describe basic results on differential privacy used to obtain our theorems.

We denote by $\mathcal{N}(\mu, \sigma^2)$ the Gaussian distribution with mean $\mu$ and standard deviation $\sigma$. The Laplace distribution with mean $0$ and standard deviation $\sqrt{2}b$ is denoted by $\mathrm{Lap}(b)$.

\begin{definition}[$(\eps, \delta)$-indistinguishability]
Two random-variables $R_1, R_2$ over the same outcome space $\mathcal{Y}$ (and $\sigma$-algebra $\Sigma_{\mathcal{Y}}$) are $(\eps, \delta)$-indistinguishable, denoted $R_1 \approx_{(\eps, \delta)} R_2$, if for all 
events $S \in \Sigma_{\mathcal{Y}},$ the following hold:
\begin{align*}
   \Pr[R_1 \in S] \leq e^{\eps} \Pr[R_2 \in S] + \delta; \\
   \Pr[R_2 \in S] \leq e^{\eps} \Pr[R_1 \in S] + \delta.
\end{align*}
\end{definition}
Next, we generalize  the definition  of event-neighboring and item-neighboring streams (\Def{neigbhors}). 
\begin{definition}[$\ell$-Neighboring streams]  Let $x, x' \in \uni_\pm^T$ be two streams of length $T$. For any $k\in\N$, streams $x$ and  $x'$  are $\ell$-event-neighbors (respectively, $\ell$-item-neighbors) if one can be obtained from the other in a sequence of $\ell$ steps, where each step replaces a stream with an event-neighboring (respectively, item-neighboring) stream.
\end{definition}

\begin{lemma}[Group privacy~\cite{DworkMNS16}]\label{lem:group_privacy} Every $(\eps, \delta)$-event-level (or item-level) DP mechanism $\mech$ is $\left(\ell \eps, \delta' \right)$-event-level (or item-level) DP for groups of size $\ell$, where $\delta' = \delta\frac{e^{\ell \eps} -1}{e^\eps-1}$; that is, for all $\ell$-event (or $\ell$-item) neighboring data streams $x, x'$, it holds that
$    \mech(x) \approx_{\ell \eps, \delta'}   \mech(x')$.
\end{lemma}

\subsection{Preliminaries on zero-concentrated differential privacy (zCDP)}\label{sec:zCDP}

This section describes \emph{zero-concentrated differential privacy (zCDP)}, a variant of differential privacy that is less stringent than pure differential privacy, but more stringent than approximate differential privacy. Using this notion of privacy, one can show tight bounds for the Gaussian mechanism and cleaner and tighter bounds for composition. In contrast to $(\eps,\delta)$-differential privacy, zCDP requires output distributions on all pairs of neighboring datasets to be $\rho$-close (\Def{rho-indistinguishable}) instead of $(\eps,\delta)$-indistinguishable.

\begin{definition}[R\'enyi divergence \cite{Renyi61}]
Let $Q$ and $Q'$ be distributions on $\mathcal{Y}$. For $\xi \in (1,\infty)$, the R\'enyi divergence of order $\xi$ between $Q$ and $Q'$(also called the $\xi$-R\'enyi Divergence) is defined as
\begin{align}
    D_{\xi}(Q \| Q') = \frac{1}{\xi-1} \log\left( \E_{r \sim Q'} \left[ \left(\frac{Q(r)}{Q'(r)}\right)^{\xi-1} \right]  \right).
\end{align}
Here $Q(\cdot)$ and $Q'(\cdot)$ denote either probability masses (in the discrete case) or probability densities (when they exist). More generally, one can replace  $\frac{Q(.)}{Q'(.)}$ with the the Radon-Nikodym derivative of $Q$ with respect to $Q'$.
\end{definition}

\begin{definition}[$\rho$-Closeness]\label{def:rho-indistinguishable}
Random variables $R_1$ and $R_2$ over the same outcome space $\mathcal{Y}$ are  {\em $\rho$-close} (denoted $R_1 \simeq_{\rho} R_2$) if for all $\xi \in (1,\infty)$, 
\begin{align*}
D_{\xi}(R_1\|R_2) \leq \xi\rho \text{ and }  D_{\xi}(R_2\|R_1) \leq \xi\rho,
\end{align*}
where $D_{\xi}(R_1\|R_2)$ is the $\xi$-R\'enyi divergence  between the distributions of $R_1$ and $R_2$.
\end{definition}

\begin{definition}[zCDP in batch model~\cite{BunS16}]
A randomized batch algorithm $\alg : \X^n \to \mathcal{Y}$ is $\rho$-zero-concentrated differentially private ($\rho$-zCDP), if, for all neighboring datasets $y,y' \in \mathcal{X}^n$,
$$\alg(y) \simeq_{\rho} \alg(y').$$
\end{definition}

One major benefit of using zCDP is that this definition of privacy admits a clean composition result. We use it when analysing the privacy of the algorithms in \Sec{upper-bounds}.

\begin{lemma}[Composition \cite{BunS16}] \label{lem:cdp_composition}
Let $\alg : \mathcal{X}^n \to \mathcal{Y}$ and $\alg' : \mathcal{X}^n \times \mathcal{Y} \to \mathcal{Z}$ be batch algorithms. Suppose $\alg$ is $\rho$-zCDP and $\alg'$ is $\rho'$-zCDP. Define batch algorithm $\alg'' : \mathcal{X}^n \to \mathcal{Y} \times \mathcal{Z}$ by $\alg''(y) = \alg'(y,\alg(y))$. Then $\alg''$ is $(\rho+\rho')$-zCDP.
\end{lemma}

\begin{lemma}[Post-processing~\cite{DworkMNS16,BunS16}]
\label{lem:postprocess} If $\alg: \mathcal{Y} \rightarrow \R^k$ is $(\eps, \delta)$-DP and $\mathcal{B} : \R^k \rightarrow \mathcal{Z}$ is any randomized function, then the algorithm $\mathcal{B} \circ \alg$ is $(\eps, \delta)$-DP.
Similarly, if $\alg$ is $\rho$-zCDP then the algorithm $\mathcal{B} \circ \alg$ is $\rho$-zCDP.
\end{lemma}

The \emph{Gaussian mechanism}, defined next, is used in \Sec{upper-bounds}. It privately estimates a real-valued function on a dataset by adding Gaussian noise to the value of the function. 

\begin{definition}[Sensitivity] Let $f: \cY \rightarrow \R^k$ be a function. Its $\ell_2$-sensitivity is defined as
$$ \max_{\text{neighbors } y, y' \in \cY} \|f(y) - f(y')\|_2.$$ To define $\ell_1$-sensitivity, we replace the $\ell_2$ norm with the $\ell_1$ norm.
\end{definition} 

\begin{lemma}[Gaussian mechanism \cite{BunS16}] \label{lem:gaussian-mech}
Let $f : \X^n \to \mathbb{R}$ be a function with $\ell_2$-sensitivity at most $\Delta_2$.
Let $\alg$ be the batch algorithm 
that, on input $y$, releases a sample from $\mathcal{N}(f(y), \sigma^2)$. Then $\alg$ is $(\Delta_2^2/(2\sigma^2))$-zCDP.
\end{lemma}

The final lemma in this section relates zero-concentrated differential privacy to $(\eps,\delta)$-differential privacy.

\begin{lemma}[Conversion from zCDP to DP \cite{BunS16}]\label{lem:CDPtoDP}
For all $\rho,\delta > 0$, if batch algorithm $\alg$ is $\rho$-zCDP, then $\alg$ is $(\rho+2\sqrt{\rho \log(1/\delta)},\delta)$-DP. Conversely, if $\cA$ is $\eps$-DP, then $\cA$ is $(\frac{1}{2}\eps^2)$-zCDP. 
\end{lemma}

\section{Item-level private mechanisms for $\countdistinct$}\label{sec:upper-bounds}

In this section, we prove a version of \Thm{item-epsdel} with zero concentrated differential privacy (zCDP). 
Using this notion of privacy, one can show tight bounds for the Gaussian mechanism and cleaner and tighter bounds for composition.
We start by proving \Thm{item-w}, and explain in \Sec{introub} how this can be used to prove \Thm{item-epsdel}. 

\begin{theorem}[Upper bound]\label{thm:item-w}
    For all $\rho \in(0,1]$ and sufficiently large $T \in \N$, there exists a $\rho$-item-level-zCDP mechanism  
    for $\countdistinct$ that is 
    $\alpha$-accurate for all turnstile streams $x$ of length $T$,  where
    $$\alpha = O\Big(\frac{\sqrt{w_x}\log T + \log^3T}{\sqrt{\rho}}\Big),$$
    and $w_x$ is the maximum flippancy of the stream $x$.
\end{theorem} 
In \Sec{item-upper}, we describe a modification to the binary-tree mechanism which, when analyzed carefully, provides the desired error guarantees---but only if the maximum flippancy of the stream is known upfront. In \Sec{item-w}, we 
use this mechanism, in conjunction with a method for adaptively estimating the flippancy bound, to obtain our item-level-DP mechanism for $\countdistinct$.

\subsection{Enforcing a given flippancy bound $w$}\label{sec:item-upper}

When a flippancy upper bound $w$ is given upfront, we leverage the structure of the binary-tree mechanism to privately output the number of distinct elements at each time $t \in [T]$, where $T$ is the stream length.  The mechanism and its error guarantees are presented in \Alg{item-upper} and \Thm{item-upper}, respectively. To provide intuition, we first describe the mechanism when it is run on streams with maximum flippancy at most $w$. We then 
discuss a modification that ensures privacy of the mechanism for all streams regardless of maximum flippancy.

\Alg{item-upper} stores vectors $\boundedexist{\ele} \in \zo^T$ for all elements $\ele \in \univ$ that appear in the stream. For streams with maximum flippancy at most $w$, the vector $\boundedexist{\ele}$ is equal to the existence vector $\existence{\ele}$. In this case,  by \Def{existence-vector}, the number of distinct elements at timestep $t \in [T]$ equals $\sum_{\ele \in \univ} \boundedexist{\ele}[t]$. The mechanism outputs values $\sum_{\ele \in \univ} \boundedexist{\ele}[t]$ with Gaussian noise added according to the binary-tree mechanism, with privacy parameter $\approx \rho/w$ (see \Def{binary_tree_rv})---that is, with noise scaled up by a factor of $\approx \sqrt{w/\rho}$. The accuracy of this mechanism follows from that of the binary-tree mechanism.

However, if the mechanism computed $\existence{\ele}$ instead of $\boundedexist{\ele}$, it would not be
private for streams with maximum flippancy greater than $w$, since it adds noise that scales according to $w$. That is because for every stream $x \in \univ_\pm^T$ with maximum flippancy $w_x > w$ there exists a neighboring stream $x'$ such that  the vectors $\countdistinct(\dstream)$ and $\countdistinct(\dstream')$ differ in as many as $\Theta(w_x)$ indices. To provide privacy for such streams, the mechanism simply ``truncates'' the vector $\existence{\ele} \in \zo^T$ to obtain $\boundedexist{\ele}[t] = 0$ for all $t \geq t^*$ if the flippancy of $\ele$ in $x[1:t^*]$ exceeds $w$. This corresponds to running the naive version of the mechanism (that uses $\existence{\ele}$ instead of $\boundedexist{\ele}$) on the ``truncated'' version of the stream $x$, where elements in $x$ are ignored after their flippancy exceeds $w$.
(Note that the computation of $\boundedexist{\ele}$ can be done 
online since $\boundedexist{\ele}[t]$ depends only on $\dstream[1:t]$.) With a careful analysis of the value stored in each node of the binary tree, we are able to show that this mechanism is $\rho$-item-level-zCDP
for all streams, however, it loses accuracy for streams with many high flippancy elements. In \Sec{item-w}, we leverage this mechanism to provide estimates of $\countdistinct$ that are both accurate and private for \emph{all} streams.

\begin{theorem}[Mechanism for a given flippancy bound $w$]\label{thm:item-upper}
    Fix $\rho \in(0,1]$, sufficiently large $T \in \N$, and $w \leq T$. \Alg{item-upper} is a mechanism for $\countdistinct$ for turnstile streams that is $\rho$-item-level-zCDP for all input streams of length $T$, and $\alpha$-accurate for streams of length $T$ with maximum flippancy at most $w$, where
   $\alpha = O\paren{\frac{\sqrt{w}\log T + \log^3T}{\sqrt{\rho}}}.$
\end{theorem}

\begin{algorithm}[ht!]
    \caption{Mechanism $\cM$ for $\mathsf{CountDistinct}$ with given flippancy bound}
    \label{alg:item-upper}
    \begin{algorithmic}[1]
        \Statex \textbf{Input:} Time horizon $T \in \N$, privacy parameter $\rho > 0$, flippancy upper 
        bound $w > 0$, stream $x \in \uni_\pm^T$
        \Statex \textbf{Output:} Vector $\out \in \R^T$ of distinct count estimates
        \State Sample a binary-tree random variable $Z \in \R^T$ with parameter $\rho' = \frac{\rho}{4w(\log T + 1)}$ \Comment{\Def{binary_tree_rv}} \label{step:choose_rho}
        
        \State Initialize $\uni_{\dstream}=\emptyset$
        
        \ForAll{$t \in [T]$}
        \State Obtain entry $x[t]$ and skip to Step~\ref{step:output-distinct-elements} if $x[t]=\bot$
        \State Suppose $x[t]$ is an insertion or deletion of a universe element $\ele$
        \State {\bf if} $\ele\notin \uni_\dstream$ {\bf then} insert $\ele$ into $\uni_\dstream$; initialize $\mathsf{count}_{\ele} = 0$ and $\boundedexist{\ele} = \mathsf{0}^T$
        \Comment{vector with $T$ zeros}
        
        \State {\bf if} $x[t] = +u$ {\bf then} $\mathsf{count}_u\mathrel{+}=1$ {\bf else} $\mathsf{count}_u \mathrel{-=}1$
        \ForAll{$v \in \uni_\dstream$}
        
        \State {\bf if} $\flip(v, x[1:t]) \leq w$ and $\mathsf{count}_{v} > 0$ {\bf then} set $\boundedexist{v}[t] = 1$

        \label{step:update_f}
        \EndFor
        \State Return $\out[t] = (\sum_{u \in \uni_x} \boundedexist{\ele}[t]) + Z[t]$
        \label{step:output-distinct-elements}
        \EndFor
    \end{algorithmic}
\end{algorithm}

\begin{definition}[Dyadic decomposition]\label{def:bintree_dyadic}
For $t \in \N$, the {\em dyadic decomposition} of the interval $(0,t]$ is a set of at most $\log t + 1$ disjoint intervals whose union is $(0, t]$, obtained as follows. Consider the binary representation of $t$ (which has at most $\log t + 1$ bits), and express $t$ as a sum of distinct powers of $2$.
Then, the first interval is $(0, r]$, where $r$ is the largest power of $2$ in the sum. The second interval starts at $r+1$ and its size is the second largest power of $2$ in the sum. The remaining intervals are defined similarly for all remaining summands. For example, for $t=11=8+2+1$, the dyadic decomposition of $(0,11]$ is the intervals $(0,8]$, $(8,10]$ and $(10,11]$. 
\end{definition}

\begin{definition}[Binary tree and binary-tree random variable] 
\label{def:binary_tree_rv}
Let $\rho > 0$ be a privacy parameter and $T \in \N$ be a power of 2.  Consider a complete binary tree with $T$ leaves whose nodes are labeled as follows. The $T$ leaves are labeled by the intervals $(t-1, t]$ for all $t \in [T]$ and the internal nodes are labeled by intervals obtained from the union of their children's intervals. Specifically, the binary tree consists of $\log T +1$ levels. A level $\ell \in [0,  \log T]$ partitions the interval $(0, T]$ into a set of $\frac{T}{2^{\ell}}$ disjoint intervals, each of length $2^{\ell}$, of the form $((i-1)\cdot 2^{\ell},i\cdot 2^{\ell}]$. The nodes in level $\ell$ are labelled by the intervals in this partition. 

The binary-tree random variable $Z \in \R^T$ with parameter $\rho$ is defined as follows. For each node $(t_1, t_2]$ in the binary tree with $T$ leaves, let $Z_{(t_1, t_2]} \sim  \cN(0, 1/\rho)$. For each $t \in [T]$, consider the dyadic decomposition of the interval $(0, t]$ (Definition~\ref{def:bintree_dyadic}) and let $Z[t]$ be the sum of the random variables corresponding to the intervals in this dyadic decomposition.
\end{definition}

\begin{proof}[Proof of \Thm{item-upper}]
 We start by reasoning about the privacy of \Alg{item-upper}. It is helpful to think about \Alg{item-upper} more explicitly in terms of the binary tree mechanism. We define a mechanism $\cM'$ that returns noisy values for all nodes of the binary tree from \Def{binary_tree_rv}
 and show that the output of  \Alg{item-upper} can be obtained by post-processing the output of $\cM'$. 

 Assume w.l.o.g.~that $T$ is a power of 2; otherwise, consider the value $2^{\lceil\log_2 T\rceil}$ instead.
 Fix a stream $x$ as the input to \Alg{item-upper}.
    For all $t\in[T]$, let $F[t] = \sum_{\ele \in \uni}\boundedexist{\ele}[t]$, where the vector $\boundedexist{\ele}$ is obtained by the end of running \Alg{item-upper} with input $x$. (If $u \notin \uni_x$, set $\boundedexist{\ele} = \mathsf{0}^T$. Set $F[0] = 0$). Define $\cM'$ so that on input $x$, for each node $(t_1,t_2]$ of the binary tree with $T$ leaves, it outputs $F[t_2] - F[t_1] + Z_{(t_1, t_2]}$. 
    
    We show how to obtain the outputs of \Alg{item-upper} from the outputs of $\cM'$. For each time step $t\in[T]$ consider the dyadic decomposition 
    of the interval $(0,t]$ into $k$ intervals $(t_0, t_1], (t_1, t_2], \dots, (t_{k-1}, t_{k}]$, corresponding to nodes in the binary tree, where $t_0 = 0$, $t_k = t$, and $k \leq \log T + 1$. Add the outputs corresponding to the nodes in the dyadic decomposition of $(0,t]$ to obtain
    \begin{align*}
        \sum_{i \in [k]} F[t_i] - F[t_{i-1}] + Z_{(t_{i-1}, t_i]} 
        = F[t_k] - F[0] + \sum_{i \in [k]} Z_{(t_{i-1}, t_i]} = F[t] + Z[t],
    \end{align*}
    where the last equality holds because $Z$ is a binary-tree random variable (see \Def{binary_tree_rv}). 
   The right-hand side is exactly the $t$-th output of \Alg{item-upper}.

    We now show that $\cM'$ is $\rho$-item-level-zCDP, which implies that \Alg{item-upper} is $\rho$-item-level-zCDP. For each level $\ell \in [0, \log T]$ of the binary tree, define a vector $G_\ell$ of 
    length $\frac{T}{2^{\ell}}$ at that level as follows: 
     \begin{align*}
         G_\ell[i] = F[i\cdot 2^{\ell}] - F[(i-1)\cdot 2^{\ell}]  \quad \text{ for all } i \in [T/2^{\ell}].  
     \end{align*}
     The random variable $G_\ell[i]$ + $Z_{(2^{\ell}\cdot(i-1), 2^\ell\cdot i]}$ equals the output of $\cM'$ for node $(2^{\ell}\cdot(i-1), 2^\ell\cdot i]$ in the binary tree. Let $G = (G_0, G_1 \dots, G_{\log T})$.  Mechanism $\cM'$ corresponds to applying the Gaussian mechanism (\Lem{gaussian-mech}) to the output vector $G$,  since the variables
     $Z_{(t_1,t_2]}$ corresponding to the nodes $(t_1,t_2]$ of the binary tree
     are independent. We now bound the $\ell_2$-sensitivity of $G$. 
     Let $\dstream'$ be an item-neighboring stream of $\dstream$, and let $\ele \in \uni$ be the universe element on which the two streams differ. Define $\tilde{f_\ele'}$, $F'$, $G_\ell'$, and $G'$ for the stream $x'$ analogously to the definitions of $\boundedexist{\ele}$, $F$, $G_{\ell}$, and $G$ for stream $x$. 
     
    \begin{lemma}[$\ell_2$-sensitivity of $G$] \label{lem:g_sensitivity}
    For all item-neighboring streams $x$ and $x'$,
            \begin{align}
        \| G - G'\|_2 
        \leq \sqrt{ 8w (\log T+1)}.
        \label{eq:sensitivity}
     \end{align}
    \end{lemma}
    \begin{proof}
        We first show that for all levels $\ell \in [0, \log T]$,
        \begin{equation*}
           \| G_{\ell} -G_{\ell}' \|_2 \leq \sqrt{8w}.
        \end{equation*}
      Fix some $\ell \in [0, \log T]$ and $i \in [\frac{T}{2^{\ell}}]$. Define $i_1 = (i-1)\cdot 2^{\ell}$ and $i_2 = i \cdot 2^{\ell}$. Since the streams $x$ and $x'$ only differ in the occurrences of element $\ele$, the values $G_\ell[i]$ and $G_\ell'[i] $ differ by at most $2$:
     \begin{align}
        | G_\ell[i] - G_\ell'[i] | =
        | \boundedexist{\ele}[i_2] -\boundedexist{\ele}[i_1] - \tilde{f_\ele'}[i_2] + \tilde{f_\ele'}[i_1] | \leq 2, \label{eq:constant_2}
     \end{align}
     where the inequality follows from the fact that $\boundedexist{\ele}, \tilde{f_\ele'} \in \{0,1\}^T$.

     Observe that $G_{\ell}[i] - G_{\ell}'[i] \neq 0$ only if at least one of the following hold: $\boundedexist{\ele}[i_1] \neq \boundedexist{\ele}[i_2]$ or $\tilde{f_\ele'}[i_1] \neq \tilde{f_\ele'}[i_2]$. Define the flippancy of a vector $a \in \R^T$, denoted $\flip(a)$, as the number of pairs of adjacent entries of $a$ with different values.
     The condition  $\boundedexist{\ele}[i_1] \neq \boundedexist{\ele}[i_2]$ implies that a ``flip'' occurs in the vector $\boundedexist{\ele}$ between indices $i_1$ and $i_2$. The same holds for $\tilde{f_\ele'}$. By the design of \Alg{item-upper} (and consequently $\cM'$), $\flip(\boundedexist{\ele}) \leq w$ and $\flip(\tilde{f_\ele'}) \leq w$.     Additionally, all intervals $(i_1, i_2]$  for a fixed $\ell$ are disjoint. Hence, the number of intervals $i \in [\frac{T}{2^{\ell}}]$ such that $G_{\ell}[i] \neq G'_{\ell}[i]$ is at most $2w$.
Combining this fact with
Equation~(\ref{eq:constant_2}), we obtain the following upper bound on the $\ell_2$-sensitivity of $G_{\ell}$ for all levels $\ell \in [0,\log T]$:
     \begin{align*}
         \| G_{\ell} - G_{\ell}' \|_2^2 = \sum_{i \in [T/2^{\ell}]} (G_\ell[i] - G_\ell'[i])^2 \leq 2w \cdot 2^2 = 8w.
     \end{align*}
Combining the equalities for all levels, we obtain
     \begin{align*}
        \| G - G'\|_2^2 = 
        \sum_{\ell \in [0,\log T]} \|G_{\ell} - G_{\ell}'\|_2^2  
        \leq 8w (\log T+1).
     \end{align*}
     This concludes the proof of \Lem{g_sensitivity}.
\end{proof}

Recall that mechanism $\cM'$ corresponds to applying the Gaussian mechanism to the output vector $G$. 
By the $\ell_2$-sensitivity bound for $G$ (\Lem{g_sensitivity}), and the 
privacy of the Gaussian mechanism (\Lem{gaussian-mech}), we obtain that $\cM'$ is $(8w(\log T+1)\rho'/2)$-zCDP, where $\rho'$ is chosen in \Step{choose_rho} of \Alg{item-upper}. By the choice of $\rho'$, mechanism $\cM'$ (and hence, \Alg{item-upper}) is $\rho$-item-level-zCDP.

Next, we analyze the accuracy of \Alg{item-upper}. Suppose the input stream $x$ has maximum flippancy at most $w$. Then the variables $\boundedexist{\ele}$ from \Alg{item-upper} with input stream $x$ satisfy  $\boundedexist{\ele} = \existence{\ele}(x)$. Recall that $\countdistinct(x) \in \R^T$ denotes the vector of distinct counts for $x$. 
Then  $\countdistinct(x) = \sum_{i\in \uni} \existence{\ele}(x) = \sum_{i\in \uni} \boundedexist{\ele}(x)= \out-Z$, where $s$ is the vector of outputs of \Alg{item-upper} defined in \Step{output-distinct-elements}. 
As a result, $\err_{\countdistinct}(x, \out) = \max_{i \in [T]} |Z[t]|$. Each $Z[t]$ is the sum of at most $\log T + 1$ independent Gaussian random variables distributed as $\cN(0, \frac{1}{\rho'})$.
Therefore, $Z[t]$ is also Gaussian with mean 0 and variance at most $\frac{\log T + 1}{\rho'}$. We bound the error of our algorithm by standard concentration inequalities for Gaussian random variables. 
Set $m =  \sqrt{16w (\log T + 1)^2/\rho}$. By \Lem{gauss_max}, 
\begin{align*}
    \Pr[ \err_{\countdistinct}(x, \out) \geq m ] = \Pr\Big[\max_{t \in [T]}Z[t] \geq m \Big] \leq 2Te^{-\frac{m^2 \rho'}{2(\log T + 1)} } = 2Te^{-2(\log T + 1)} = \frac{2}{e^2T}.
\end{align*}
Note that $\frac{2}{e^2T} \leq \frac{1}{100}$ for large enough $T$, which concludes the proof of \Thm{item-upper}.
\end{proof}

\subsection{Adaptively estimating a good flippancy bound $w$}\label{sec:item-w}

In this section, we leverage the privacy and accuracy guarantees of \Alg{item-upper} to construct a new mechanism (\Alg{item-w}) for estimating $\countdistinct$. 
It achieves the privacy and accuracy guarantees of \Thm{item-w}, when the maximum flippancy is not known upfront. 
\Alg{item-w} instantiates $\log T+1$ different copies $\mathcal{B}_0, \dots \mathcal{B}_{\log T}$ of \Alg{item-upper} with flippancy bounds $ 2^0, \dots, 2^{\log T}$, respectively (the maximum flippancy of a stream is at most $T$). To obtain an accurate estimate of the distinct elements count,
at each time $t\in[T]$, we privately select $i \in [0, \log T]$ such that the output of $\mathcal{B}_i$ satisfies the desired accuracy guarantee for the stream entries $x[1:t]$ received so far. Selecting such $i$ amounts to selecting a good bound on the maximum flippancy of the stream $x[1:t]$. Next, we describe how to obtain this bound using the sparse vector technique  (\Alg{svt}).

 The maximum flippancy has high sensitivity: changing one stream entry can drastically change the maximum flippancy. However, the number of items with flippancy greater than any particular threshold is a function of sensitivity one. Furthermore, since \Alg{item-upper} when run with flippancy bound $w$ already has error about $\sqrt{w/\rho}$, its accuracy guarantee remains asymptotically the same even if it simply ignores
that many elements with flippancy greater than $w$. 
Thus, \Alg{item-w} uses the sparse vector technique to maintain an upper bound on the flippancy of $x[1:t]$ such that not too many elements in $x[1:t]$ violate that bound. This bound, in combination with the error guarantee of \Alg{item-upper}, suffices to provide the desired low error guarantee.    
Since the sparse vector algorithm remains differentially private even when its queries are chosen adaptively, the privacy guarantees of \Alg{item-w} follow from the privacy of 
Algorithms~\ref{alg:item-upper}~and~\ref{alg:svt}.

\begin{algorithm}[ht] 
    \caption{$\mathsf{SVT}$: Answering Threshold Queries with Sparse Vector Technique}
    \label{alg:svt}
    \begin{algorithmic}[1]
        \Statex \textbf{Input:} Stream $x$, queries $q_1, q_2, \dots $ of sensitivity $1$, 
        cutoff $c > 0$, privacy parameter $\rho$
        \Statex \textbf{Output:} Stream of $\abov$ or $\below$ answers
        \State Let $\eps = \sqrt{2\rho}$ and set $\mathsf{count} = 0$
        \State Let $Z \sim \mathrm{Lap}(2/\eps)$
        \For{each query $q_t$}
        \State Let $Z_t \sim \mathrm{Lap}(4c/\eps)$
        \If{$q_t(x) + Z_t \geq 
        Z$ and $\mathsf{count} < c$}
        \State Return $\abov$ 
        \State $\mathsf{count} = \mathsf{count} + 1$ \label{step:update-count}
        \Else{}
        \State Return $\below$
        \EndIf
        \EndFor
    \end{algorithmic}
\end{algorithm}

\begin{algorithm}[ht]
    \caption{Mechanism $\cM$ for $\countdistinct$}
    \label{alg:item-w}
    \begin{algorithmic}[1]
        \Statex \textbf{Input:} Time horizon $T \in \N$, privacy parameter $\rho > 0$, stream $x \in \uni_\pm^T$
        \Statex \textbf{Output:} Vector $\out$ of distinct count estimates
        \State Initialize vector $w_{\max} = 1 \circ 0^{T-1}$
        \ForAll{$i \in [0, \log T]$}
            \State Initialize $\cB_i$ as \Alg{item-upper} with horizon $T$, privacy parameter $\frac{\rho}{2(\log T+1)}$, flippancy $2^i$ \label{step:b_i}
        \EndFor
        \State Initialize $\mathsf{SVT}$ with privacy parameter $\rho/2$ and cutoff $\log T$
        \Comment{See \Alg{svt}}
        \ForAll{$t \in [T]$} 
         \State Obtain entry $x[t]$
        \State If $t \geq 2$, set $w_{\max}[t] = w_{\max}[t-1]$
        \ForAll{$i \in [0, \log{T}]$} 
        \State Send $x[t]$ to mechanism $\mathcal{B}_i$ and get output $\out_{i,t}$
        \EndFor
        \While{$\mathrm{True}$}
            \State Consider query $q_t = |\{u \in \uni : \flip(u, x[1:t]) \geq w_{\max}[t] \}| - \sqrt{\frac{w_{\max}[t]}{\rho}}$ \label{step:query}
            \State Send query $q_t$ to $\mathsf{SVT}$ and if the output is ``$\below$'', $\mathbf{break}$
            \State Update $w_{\max}[t] = 2 \cdot w_{\max}[t]$ \label{step:item-w-update}
        \EndWhile
        \State Return $\out_{j, t}$ for $j = \log (w_{\max}[t])$\label{step:item-w-output} \Comment{Note that $j \in [0, \log T]$ }
        \EndFor
    \end{algorithmic}
\end{algorithm}

The accuracy and privacy guarantees of the sparse vector technique (\Alg{svt}) are stated in \Thm{svt}. 

\begin{definition}[$\gamma$-accuracy \cite{DworkNRRV09}]  Let  $(a_1, \dots, a_k) \in$  $\{ \abov, \below\}^k$ be a vector of answers in response to $k$ queries $q_1, \dots, q_k$ on a dataset $x$. We say $(a_1, \dots, a_k)$ is $\gamma$-accurate 
if $q_t(x) \geq 
- \gamma$ for all $a_t  = \abov$ and $q_t(x) \leq 
\gamma$ for all $a_t =\below$.
 \end{definition}

\begin{theorem}[\cite{DworkNRRV09, LyuSL17}]
    \label{thm:svt}
    \Alg{svt} is $\sqrt{2\rho}$-DP (and therefore $\rho$-zCDP). Let $k$ be the index of the last ``$\abov$'' query answered by \Alg{svt} (before cutoff $c$ has been crossed). For all $\beta \in (0,1)$, with probability at least $1-\beta$, the vector of answers to the queries $q_1, \dots, q_k$ is $\gamma$-accurate for $\gamma = \frac{8c(\ln k + \ln(2c/\beta))}{\sqrt{2\rho}}$.
\end{theorem}

To prove \Thm{item-w}, we use a slightly stronger result (\Cor{item-upper}) on the accuracy of \Alg{item-upper}.

\begin{corollary}\label{cor:item-upper}
Fix $\rho > 0$, sufficiently large $T \in \N$, and a flippancy bound $w \leq T$. \Alg{item-upper} satisfies the following accuracy guarantee for all streams $\dstream \in \univ_\pm^T$ and $t\in[T]$: if at most $\ell$ elements in the prefix $\dstream[1:t]$ of the stream $x$ have flippancy greater than $w$, then,  with probability at least $1-\frac{1}{T}$, \Alg{item-upper} has error $O(\ell + \sqrt{\frac{w \log^2 T}{\rho}})$ over all time steps from $1$ to $t$. 
\end{corollary}
\begin{proof}
    The proof is similar to the accuracy analysis in \Thm{item-upper}, once we observe that $\countdistinct(x) \leq \ell \cdot \mathsf{1}^T + \sum_{\ele \in \uni} \tilde{f}_i(x)$, where $\mathsf{1}^T$ is a vector of length $T$.
\end{proof}

We are now ready to prove \Thm{item-w}.

\begin{proof}[Proof of \Thm{item-w}]
    We start by showing that \Alg{item-w} is $\rho$-item-level-zCDP. \Alg{item-w} accesses the stream $x$ via \Alg{svt} and the algorithms $\cB_i$ for $i \in [0,\log T]$ (instantiations of \Alg{item-upper}).
    By \Thm{item-upper}, \Alg{item-upper} with privacy parameter $\rho'$ is $\rho'$-item-level-zCDP. Since we use $(\log T +1)$ instantiations of \Alg{item-upper}, each with privacy parameter $\frac{\rho}{2(\log T+1)}$, by composition, the aggregate of the calls to \Alg{item-upper} is $(\frac{\rho}{2})$-item-level-zCDP. 
    We now show that the aggregate of the calls to \Alg{svt} is $(\frac{\rho}{2})$-item-level-zCDP. Note that the queries $q_t$ for $ t\in[T]$ considered in \Step{query} of \Alg{item-w} have sensitivity $1$ for item-neighboring streams (the number of items with flippancy above a certain threshold can change by at most 1 for item-neighboring streams). 
    By \Thm{svt}, the aggregate of the calls to \Alg{svt} is $(\frac{\rho}{2})$-item-level-zCDP. Another invocation of the composition lemma gives that \Alg{item-w} is $\rho$-item-level-zCDP.

    We now analyze the accuracy of \Alg{item-w}. Set $\beta_{\msf{SVT}} = 0.005$, $k=T$, $c=\log T$, and $\gamma_{\msf{SVT}} = \frac{8\log T(\log T + \log(400\log T)) }{\sqrt{2\rho}}$. 
    Let $E$ be the event that the vector of answers output by the sparse vector algorithm (\Alg{svt}) until the cutoff point $\log T$ is $\gamma_{\msf{SVT}}$-accurate. 
    By \Thm{svt}, $\Pr[E] \geq 0.996$. We condition on $E$ for most of the following proof. 
    
    Set $t^*_{-1} = 1$. Let $t^*_i$ be the last time step at which the output of instance $\mathcal{B}_i$ is used as the output of \Alg{item-w}. Instance $\mathcal{B}_i$ of \Alg{item-upper} is run with parameter $w= 2^i$. Conditioned on event $E$, its outputs are used only at times $t^*_{i-1}< t \leq t^*_i$ when at most $\ell_i = \BigO{\frac{\log^2 T}{\sqrt{\rho}}} + \sqrt{\frac{2^{i}}{\rho}}$ elements have flippancy greater than $2^i$. 
    By \Cor{item-upper}, with probability at least $1-\frac{1}{T}$, the error of $\mathcal{B}_i$ over time steps $t^*_{i-1},\dots,t^*_i$ is 
    $$\BigO{\frac{\log^2 T + \sqrt{2^i \log^2 T}}{\sqrt{\rho}}}.$$
    
    Since  exactly $(\log T+1)$ instances of \Alg{item-upper}  are run within \Alg{item-w}, a union bound over the failure probability of each of those instances gives us the following: Conditioned on event $E$, with probability at least $1-\frac{\log T+1}{T}$, the error of \Alg{item-w} over time steps $t\in[T]$ is
    \begin{align}\label{eq:acc-bound}
    \BigO{\frac{\log^2 T + \sqrt{w_{\max}[t] \log^2 T}}{\sqrt{\rho}}}.
    \end{align}
 
    This bound on the error holds with probability $1-\frac{\log T + 1}{T} \geq 0.995$ for sufficiently large $T$.
   
    \begin{claim}\label{clm:maxflip-ub}
    Let $w_t$ be the (true) maximum flippancy of the sub-stream $x[1:t]$, consisting of the first $t$ entries of the input stream $x\in \univ_\pm^T$ to \Alg{item-w}. Then, for all $t \in [T]$, when the algorithm reaches Step~\ref{step:item-w-output}, 
    $$\wmax[t] \leq \max(2w_t,2\rho\gamma^2_{\msf{SVT}}).$$
    \end{claim}
    
    \begin{proof} We consider two cases.
  
  \paragraph{(Case 1)} $t \in [T]$ during which $\mathsf{count} < c$ for \Alg{svt}.

   Let $z$ be the value of $\wmax[t]$ when \Alg{item-w} reaches Step~\ref{step:item-w-output}. If $z=1$ then $z = \wmax[t] \leq 2\gamma^2_{\msf{SVT}}$ since $T > 1, \rho < 1$. So, instead assume that $z\geq 2$. 
   Let $t^* \leq t$ be the time step where $w_{\max}[t^*]$ is doubled from $\frac{z}{2}$ to $z$ during an execution of Step~\ref{step:item-w-update} of the $\mathbf{while}$ loop.
   This only happens if \Alg{svt} outputs "$\abov$" for the following query: 
    $$\Big|\Big\{\ele \in \uni \colon \flip(u, x[1:t^*]) \geq \frac{z}{2} \Big\}\Big| - \sqrt{\frac{z}{2\rho}}.$$
    If at this point $\frac{z}{2} \leq w_{t^*}$, then $\frac{z}{2} \leq w_{t}$ (because $w_t \geq w_{t^*}$.) Otherwise $\frac{z}{2} > w_{t^*}$ and therefore $|\{\ele \in \uni \mid \flip(u, x[1:t^*]) \geq \frac{z}{2} \}| = 0$. In this case, by applying \Thm{svt}, we get that $0 - \sqrt{\frac{z}{2\rho}} \geq -\gamma_{\msf{SVT}}$, which implies that $z \leq 2\rho\gamma^2_{\msf{SVT}}$.
    
    \paragraph{(Case 2)} $t \in [T]$ during which $\mathsf{count} \geq c$ for \Alg{svt}.
    
    Suppose there is some $t \in [T]$ during which $\mathsf{count} \geq  c$. Consider the last time step $t^* \in [T]$ when Step~\ref{step:update-count} of \Alg{svt} is run (for this time step, $\mathsf{count} = c -1$). 
    At this time step, $w_{max}[t^*]$ doubles from $\frac{T}{2}$ to $T$, after which it never changes again. 
    By case (1), we have that  $w_{max}[t^*] = T \leq \max(2w_{t^*},2\gamma^2_{\msf{SVT}})$. 
    Since for all $t \geq t^*$  it holds $w_{t^*} \leq w_t$ and  $w_{max}[t] = w_{max}[t^*]$, then $w_{max}[t] \leq \max(2w_{t},2\rho\gamma^2_{\msf{SVT}})$ for all $t \geq t^*$. 
  This concludes the proof of \Clm{maxflip-ub}.
    \end{proof}
    
    Now we substitute the upper bound on $w_{max}[t]$ from~\Clm{maxflip-ub} into Equation~(\ref{eq:acc-bound}) and apply $w_t \leq w$.
    We get that, for sufficiently large $T$, conditioned on event $E$, with probability at least $0.995$, the maximum error of \Alg{item-w} over time steps $t\in[T]$ is
   \begin{align}
    \BigO{\frac{\log^2 T + \sqrt{\max(w,2\rho\gamma^2_{\msf{SVT}})\log^2 T}}{\sqrt{\rho}}} 
&=\BigO{\frac{\sqrt{\max\left(\log^6 T,\;  2w\log^2 T\right)}}{\sqrt{\rho}}}.\label{eq:item-final}
    \end{align}
    Finally, by a union bound over the event $E$ and the event that the error of \Alg{item-w} is greater than Equation~(\ref{eq:item-final}) we obtain: For sufficiently large $T$, the maximum error of \Alg{item-w} over time steps $t\in[T]$ is bounded by Equation~(\ref{eq:item-final}) with probability at least $0.99$.
\end{proof}

\subsection{Proof sketch of \Thm{item-epsdel}}\label{sec:introub}

In this section, we sketch how to complete the proof of \Thm{item-epsdel} using \Thm{item-w} together with a result of~\citet{JainRSS23} on mechanisms for estimating functions of sensitivity at most $1$ in the continual release model.

\begin{theorem}[Mechanism for sensitivity-1 functions~\cite{JainRSS23}]\label{thm:recompute_rho} 
Let $f \colon \uni_{\pm}^* \to \R$ be a function of $\ell_2$-sensitivity at most $1$. Define $F \colon \uni_{\pm}^T \to \R^T$ so that  $F(x) = [f(x[1:1]), \dots, f(x[1:T])]$.  
Fix $\rho \in (0,1]$ and sufficiently large $T \in \N$. Then, there exists a $\rho$-item-level-zCDP mechanism for estimating $F$ in the continual release model that is $\alpha$-accurate where  $\alpha  = \BigO{\min\left\{ \sqrt[3]{\frac{T\log T}{\rho}},T \right\}}.$
\end{theorem}


Note that $\countdistinct(x)[t]$ has $\ell_2$-sensitivity one for item-neighboring streams for all $t \in [T]$. Let $\cM'$ be the mechanism from \Thm{recompute_rho}.
Then $\cM'$ can be used for estimating $\countdistinct$ under continual release for turnstile streams with the error guarantee stated in \Thm{recompute_rho}. 
When the maximum flippancy of the stream is larger than roughly $\rho^{1/3}T^{2/3} $, 
the mechanism $\cM'$ achieves better error than that of \Thm{item-w} (and it achieves worse error when the maximum flippancy of the stream is smaller than this threshold). A simple modification of \Alg{item-w} can get the best of both worlds---instead of having base mechanisms $\cB_0,\dots,\cB_{\log T}$ that each run \Alg{item-upper} with different flippancy parameters as input, we only have
$k+2 = \mathsf{min}(\rho^{1/3}T^{2/3}, T)$
base mechanisms $\cB_0,\dots, \cB_{k+1}$. Out of these, $\cB_0,\dots,\cB_k$ run \Alg{item-upper}, whereas $\cB_{k+1}$ runs $\cM'$. This modified algorithm has error $O\left(\mathsf{min}\left(\sqrt{\frac{w}{\rho}} \polylog T , \sqrt[3]{\frac{T\log T}{\rho}}, T\right)\right)$. The proof is similar to the proof of \Thm{item-w},  
with the only difference that for analyzing the error of base mechanism $\cB_{k+1}$ we use the error guarantee of the recompute-mechanism $\cM'$. 

Finally, \Thm{item-epsdel} follows by invoking the conversion from zCDP to approximate DP (\Lem{CDPtoDP}), and setting $\rho = \frac{\eps^2}{16\log(1/\delta)}$.

\section{Event-level privacy lower bound}\label{sec:eventlb}

In this section, we prove \Thm{insertion+deletion+eps}, providing a strong lower bound on the parameter $\alpha$ for every $\alpha$-accurate, \emph{event-level} differentially private mechanism for $\countdistinct$ in the continual release model for turnstile streams. This lower bound is parameterized by $w$, the maximum flippancy of the input stream.

\subsection{Reduction from $\innerproducts$}

We obtain our lower bound by showing that every mechanism for $\countdistinct$ for turnstile streams can be used to obtain an algorithm with similar accuracy guarantees for $\innerproducts$, the problem of estimating answers to inner product queries in the batch model. 
The reduction from $\innerproducts$ to $\countdistinct$ combines two ideas:  one is the sequential embedding technique introduced by~\citet{JainRSS23} to prove lower bounds in the continual release model and the other is a connection between the inner product of two vectors and the number of distinct elements in the concatenation of two corresponding streams. 
The latter idea was used by~\citet{MirMNW11} to give lower bounds for pan-private algorithms for counting the number of distinct elements. 
The reduction is presented in \Alg{innerprodred}. 
With this reduction, we then use previously established lower bounds on accuracy for $\innerproducts$ \cite{DinurN03,DworkMT07,MirMNW11,De12} to obtain our lower bound on $\countdistinct$. We start by proving \Lem{innerprod_reduction} (the reduction from $\innerproducts$ to $\countdistinct$). In \Sec{complete_lb}, we use \Lem{innerprod_reduction} to complete the proof of \Thm{insertion+deletion+eps}.

\Alg{innerprodred} crucially uses the following connection between the inner product of two vectors and the number of distinct elements in the concatenation of two corresponding streams. 
\begin{definition}[Stream indicator]
For a stream $\dstream\in \uni_\pm^T,$ 
let $h_{\dstream}$ represent the $0/1$ vector of length $|\uni|$, where a component $h_{\dstream}[u]=1$ iff element $u\in\uni$ has a positive count at the end of the stream.
\end{definition}
\begin{remark}\label{rem:countinner}
For every pair of insertion-only streams $\dstream$ and $\dstream'$,
$$ \langle h_{\dstream}, h_{\dstream'} \rangle = \|h_{\dstream}\|_0 + \|h_{\dstream'}\|_0 - \|h_{\dstream \circ \dstream'}\|_0 ,$$
where $\circ$ denotes concatenation and $\|.\|_0$ is the $\ell_0$ norm. Note that $\|h_{\dstream}\|_0$ is equal to the number of distinct elements in the stream $\dstream$.
\end{remark}

\begin{algorithm}[h]
    \caption{Reduction $\alg$ from $\innerproducts$ to
    $\countdistinct$}
    \label{alg:innerprodred}
    \begin{algorithmic}[1]
        \Statex\textbf{Input:} Dataset $\dset = (\dset[1],\dots,\dset[n]) \in \zo^n$, black-box access to mechanism $\mech$ for $\countdistinct$ in turnstile streams, and
        query vectors $q^{(1)},\dots,q^{(k)} \in \zo^n$
        \Statex\textbf{Output:} Estimates of inner products  $b = (b[1],\dots,b[k]) \in \R^k$
        \State Define the universe $\mathcal{U} = [n]$ 
        \State Initialize stream $z^{(0)} = \bot^n$
        \ForAll{$i \in [n]$}
        \State If $\dset[i]=1$ set $z^{(0)}[i] = +i$ \label{step:construct_x1}
        \EndFor
        \State  Initialize streams $z^{(1)}=\perp^{2n},\dots,z^{(k)}=\perp^{2n}$ and a vector $r$ of length $(2k+1)n$
        \ForAll{$(i,j) \in [n] \times [k]$ such that $q^{(j)}[i] = 1$}
        \State  Set $z^{(j)}[i] = +i$ and $z^{(j)}[n+i] = -i$ 
        \EndFor
        \State Run $\cM$ on the stream $\dstream \gets z^{(0)} \circ z^{(1)} \circ  z^{(2)} \circ \dots \circ z^{(k)} $ and record the answers as vector $r$ \label{step:construct_x_event}
        \ForAll{$j\in[k]$}
        \State Compute $\|q^{(j)}\|_0$ and let $b[j] = \|q^{(j)}\|_0 + r[n] - r[2jn]$ \label{step:innerprod_setoutput}
        \EndFor
        \State Return the estimates $(b[1],\dots,b[k])$
    \end{algorithmic}
\end{algorithm}

\begin{definition}[Accuracy of a batch algorithm for inner products]\label{def:ip-accuracy}
Let $k,n\in \N$. A randomized algorithm $\alg$ is \emph{$\alpha$-accurate} for $\innerproducts_{k,n}$ if, for all queries $q^{(1)},\dots,q^{(k)} \in \{0,1\}^n$, and all 
datasets $\dset \in \{0,1\}^n$, it outputs $b = (b[1], \dots, b[k])$ such that
$$\quad \Pr_{\text{coins of }\alg}\Big[ \max_{j \in [k]} |b[j] - \langle q^{(j)}, \dset \rangle | \leq \alpha\Big] \geq 0.99.$$
\end{definition}

We now show that if the input mechanism $\mathcal{M}$ to \Cref{alg:innerprodred} is accurate for $\countdistinct$, then \Cref{alg:innerprodred} is accurate for $\innerproducts$.

\begin{lemma}\label{lem:innerprod_reduction}
Let $\alg$ be \Alg{innerprodred}. For all  $\eps > 0, \delta \in [0,1), \alpha \in \R^{+}$ and $n,T,k \in \N$, where $T \geq (2k+1)n$,  if  mechanism $\mech$  is $(\eps, \delta)$-event-level-DP and $\alpha$-accurate for $\mathsf{CountDistinct}$ for streams of length $T$ with maximum flippancy at most $2k$, then batch algorithm~$\alg$ is $(\eps, \delta)$-DP and $2\alpha$-accurate for $\innerproducts_{k,n}$. 
\end{lemma}

\begin{proof}[Proof of \Lem{innerprod_reduction}]
Algorithm $\mathcal{A}$ is $(\eps,\delta)$-event-level-DP because $\cM$ is $(\eps,\delta)$-event-level-DP and changing a
record of the dataset $\dset$ corresponds to changing a single entry of the stream $\dstream$, and more specifically, an entry of the stream $z^{(0)}$ constructed in \Step{construct_x1} of \Alg{innerprodred}.  

We are left to prove the accuracy of $\cA$. Fix queries $q^{(1)},\dots,q^{(k)} \in \{0,1\}^n$. First, observe that $z^{(0)}$ is constructed so that $\dset$ is its stream indicator vector. Similarly, observe that for all $j \in [k]$ , the stream $z^{(j)}$ is constructed so that $q^{(j)}$ is the indicator vector for $z^{(j)}[1:n]$, namely, the first half of $z^{(j)}$.

Next, since at time $2jn$, all of the stream entries pertaining to earlier queries $q^{(1)}, \dots, q^{(j-1)}$ have been deleted and those pertaining to $q^{(j)}$ have been inserted, $\| h_{\dstream[1:2jn]} \|_0 = \|h_{z^{(0)} \circ z^{(j)}[1:n]} \|_0$ for $j\in[k]$.

The streams $z^{(0)}$ and $z^{(j)}[1:n]$ for $j \in [k]$ are all insertion-only streams. By Remark~\ref{rem:countinner},
\begin{align*}
    \langle h_{z^{(0)}}, h_{z^{(j)}[1:n]} \rangle = \|h_{z^{(0)}}\|_0 + \|h_{z^{(j)}[1:n]}\|_0 - \|h_{z^{(0)} \circ z^{(j)}[1:n]} \|_0. 
\end{align*}
As observed earlier, $h_{z^{(0)}} = \dset$, $h_{z^{(j)}[1:n]} = q^{(j)}$, and $\| h_{\dstream[1:2jn]} \|_0 = \|h_{z^{(0)} \circ z^{(j)}[1:n]} \|_0$. Thus,
\begin{align}\label{eq:innerprod-countdist}
\langle \dset, q^{(j)} \rangle & =  \|h_{z^{(0)}}\|_0 + \|q^{(j)}\|_0 - \| h_{\dstream[1:2jn]} \|_0 .
\end{align}

Finally, the constructed stream $x$ has maximum flippancy at most $2k$. To see this, note that the universe elements $i \in [n]$ such that $y[i] = 1$ always have count at least $1$ in $x[1:t]$ for all $t \in [(2k+1)n]$. The elements $i \in [n]$ such that $y[i] = 0$ are inserted and deleted at most once for each stream $z^{(j)}, j \in [k]$, and thus have flippancy at most $2k$ in the stream $x$.

Since the mechanism $\mech$ for $\countdistinct$ is $\alpha$-accurate on the constructed stream then, with probability at least $0.99$, the answers of $\mech$ are within additive error $\alpha$ of the distinct counts of the corresponding stream prefixes. Condition on this event for the rest of this proof. Then, $|r[n] - \|h_{z^{(0)}}\|_0| \leq \alpha$. Similarly, $|r[2jn] - \|h_{\dstream[1:2jn]}\|_0 | \leq \alpha$ for all $j \in [k]$. Additionally, $\|q^{(j)}\|_0$ is computed exactly by $\cA$. 
Hence, by the triangle inequality, Equation~\ref{eq:innerprod-countdist}, and the setting of $b[j]$ in Step~\ref{step:innerprod_setoutput}, we have that
$|b[j] - \langle \dset, q^{(j)} \rangle| \leq 2\alpha$ for all $j \in [k]$. 
Hence, with probability at least $0.99$, all of the estimates $b[j]$ returned by $\cA$ are within $2\alpha$ of the inner products $\langle q^{(j)}, \dset \rangle$, and so $\cA$ is $2\alpha$-accurate for $\innerproducts_{k,n}$.
\end{proof}

\subsection{From the reduction to the accuracy lower bound} \label{sec:complete_lb}

In this section, we use \Lem{innerprod_reduction} together with a known lower bound on the accuracy of private mechanisms for answering inner-product queries to complete the proof of \Thm{insertion+deletion+eps}. 

 We use the following lower bound on inner product queries. Like a similar lower bound of~\citet{MirMNW11}, this theorem is proved using the reconstruction attacks of \citet{DinurN03} and \citet{DworkMT07}, together with the argument of De~\cite{De12} that rules out reconstruction from the outputs of $(\eps,\delta)$-DP algorithms.
 
\begin{theorem}[Inner product queries lower bound (based on \cite{DinurN03,DworkMT07,MirMNW11,De12})]
\label{thm:IPlowerbound} 
There are constants $c_1 \geq 1$ and $c_2 > 0$ such that, for all sufficiently large $n>0$: if $\alg$ is $\alpha$-accurate for $\innerproducts_{k,n}$ (\Def{ip-accuracy}) with $k=c_1 n $ and $\alpha = c_2\sqrt{n}$, then $\alg$ is not $(1, \frac{1}{3})$-differentially private. 
\end{theorem}

We first prove \Thm{insertion+deletion+eps} for $\eps=1$ and then boost it to arbitrary $\eps\leq 1$ using the reduction in \Thm{epsred}. 

\begin{lemma}\label{lem:insertion+deletion}
For all $\delta  \in (0,1)$, sufficiently large $w,T \in \N$ such that $w \leq T$, and all $(1, \delta)$-event-level-DP
mechanisms that are $\alpha$-accurate for $\mathsf{CountDistinct}$ on turnstile streams of length $T$ with maximum flippancy at most $w$, if $\delta = o(\frac{1}{T})$, then
$$\alpha = \Omega(\mathsf{min}(\sqrt{w}, T^{1/4})).$$
\end{lemma}
 
\begin{proof}[Proof of \Lem{insertion+deletion}]
Fix sufficiently large $w$ such that $w \leq \sqrt{T}$. Let 
$c_1 \geq 1$ and $c_2 > 0$ 
be the constants from \Thm{IPlowerbound}.
Assume that $\mech$ is a $\left(1,o(\frac{1}{T})\right)$-event-level-DP, $(\frac{c_2}{2}\sqrt{w})$-accurate mechanism for $\countdistinct$ for turnstile streams of length $T$ with maximum flippancy at most $w$. 
Set $k = \frac{w}{2}$ and $n = \frac{k}{c_1} = \frac{w}{2c_1}$.
This choice of $k$ and $n$ satisfies the conditions of \Lem{innerprod_reduction} since the flippancy of the stream is at most $w=2k$. Moreover, for $w \leq \sqrt{T}$ we have that $(2k+1)n = (w+1)\frac{w}{2c_1} \leq w^2 \leq T$. 
Therefore, $\alg$ (\Alg{innerprodred}) is $\left(1,o(\frac{1}{T})\right)$-DP and
$(c_2\sqrt{w})$-accurate for $\innerproducts_{k,n}$.
Since $\frac{1}{T} \leq \frac{1}{n}$ and $w=O(n)$, we get that $\alg$ is $(1,o(\frac{1}{n}))$-DP and
$c_2\sqrt{n}$-accurate for $\innerproducts_{k,n}$, where $k = c_1n$.

However, by \Thm{IPlowerbound}, $\alg$ cannot be $(1,\frac{1}{3})$-differentially private. We have obtained a contradiction. Thus, the mechanism $\cM$ with the desired accuracy of $O(\sqrt{w})$ does not exist.
When $w=\sqrt{T}$, this argument gives a lower bound of $T^{1/4}$ on the accuracy of $\mech$, and this lower bound applies to all larger $w$, since a mechanism that is $\alpha$-accurate for streams with maximum flippancy at most $w > w'$ is also $\alpha$-accurate for streams with maximum flippancy at most $w'$.
\end{proof}

Finally, we invoke the reduction in \Thm{epsred} to improve the dependence on $\eps$ and complete the proof of \Thm{insertion+deletion+eps}.

\begin{proof}[Proof of \Thm{insertion+deletion+eps}]
Suppose $\eps < \frac{2}{T}$. For these values of $\eps$, we prove an error lower bound of $\Omega(T)$ via a group privacy argument. Suppose for the sake of contradiction that $\alpha \leq T/4$. Consider universe $\uni = [T]$. 
 Let $\dstream=\bot^T$ and $\dstream'$ be a stream of length $T$ such that $\dstream[t] = t$ for all $t \in [T]$. These data streams differ in $T$ entries. 
Let $r[T]$ and $r'[T]$ be the final outputs of $\mech$ on input streams $\dstream$ and $\dstream'$, respectively. 
By the accuracy of $\mech$, we have $\Pr[r[T] \leq T/4]\geq 0.99.$ Applying \Lem{group_privacy} on group privacy with $\eps \leq {2}/{T}$ and group size $\ell=T$, we get 
$\Pr[r'[T] > T/4] \leq e^{2} \cdot\Pr[r[T]> T/4] + \frac{2\delta}{\eps} \leq e^2 \cdot 0.01 + o(\frac{1}{T})
< 0.99$
for sufficiently large $T$. 
But $\countdistinct(\dstream')=T$, so
$\mech$ is  not $T/4$-accurate for $\dstream'$, a contradiction. Hence, $\alpha = \Omega(T)$.

Next, suppose $\eps \geq \frac{2}{T}$. 
\Lem{insertion+deletion} provides a lower bound of $\alpha' = \Omega\Big(\mathsf{min} \left(\sqrt{w}, T^{1/4}\right)\Big)$ for ($\eps' = 1$,$\delta' = o(1/T)$)-event-level-DP, $\alpha'$-accurate mechanisms for $\countdistinct$ on turnstile streams of length $T$ with maximum flippancy at most $w$.
By invoking \Thm{epsred}, we obtain the following lower bound on accuracy for $(\eps, \delta)$-DP mechanisms where 
$\delta = \frac{\delta'\eps}{2} = o(\frac{\eps}{T})$: 
   \begin{align*}
    \alpha = \frac{1}{\eps}\Omega\left(\sqrt{w}, (\eps T)^{1/4}) \right)
     = \Omega\left(\mathsf{min}\left(\frac{\sqrt{w}}{\eps},\frac{T^{1/4}}{\eps^{3/4}} \right)\right).
    \end{align*}
Overall, since for different parameter regimes, we get lower bounds $\Omega(T)$ and $\Omega\left(\mathsf{min}\left(\frac{\sqrt{w}}{\eps},\frac{T^{1/4}}{\eps^{3/4}} \right)\right)$, our final result is a lower bound of
$\Omega\left(\mathsf{min}\left(\frac{\sqrt{w}}{\eps},\frac{T^{1/4}}{\eps^{3/4}}, T \right)\right)$.
\end{proof}

\section{Item-level privacy lower bound}\label{sec:itemlb}

In this section, we prove \Thm{LB_CD_item} that provides strong lower bounds on the accuracy 
of any 
\emph{item-level} differentially private mechanism for $\countdistinct$ in the continual release model for turnstile streams. This lower bound is parameterized by $w$, the maximum flippancy of the input stream.

\subsection{Reduction from $\marginals{}$}
 
To prove our lower bounds for $\countdistinct$, we reduce from the problem of approximating marginals in the batch model. 

\begin{definition}[Marginals]
    The function $\marginals{n,d}: \{0,1\}^{n\times d} \to [0,1]^d$ maps a dataset $\dset$ of $n$ records and $d$ attributes to a vector $(q_1(\dset),\dots,q_d(\dset)),$ where $q_j$, called the $j^{th}$ marginal, is defined as $q_j(\dset) = \frac{1}{n}\sum_{i=1}^n \dset[i][j].$
\end{definition}

The reduction is presented in \Alg{CD_item_reduction}. 
The privacy and accuracy guarantees of our reduction are stated in \Lem{CD_item_reduction}. In \Sec{complete_lb_item}, we use \Lem{CD_item_reduction} to complete the proof of \Thm{LB_CD_item}. 

\paragraph{Overview of the reduction.} Let $\mech$ be an $(\eps,\delta)$-DP and $\alpha$-accurate mechanism for $\countdistinct$ in the \CR. We use $\mech$ to construct a $(O(\eps), O(\delta))$-DP batch algorithm $\alg$ that is $(\frac{\alpha}{n})$-accurate  for $\marginals{n,d}$. 
Consider a universe 
$\univ = [n] \cup \{\perp\}$ for $\countdistinct:\univ_\pm^T \to \N$.
The main idea in the construction (presented in \Alg{CD_item_reduction}) is to force $\mech$ to output an estimate of the marginals, one attribute at a time. 
Given a dataset $y \in \{0,1\}^{n\times d}$, the estimation of each marginal proceeds in two phases:
\begin{itemize}
    \item In \emph{phase one}, 
     $\alg$ sends element $i$ to $\mech$ for each record $\dset[i]$ with a $1$ in the first attribute. The answer produced by $\mech$ at the end of \emph{phase one} is an estimate of the sum of the first attribute of all records $\dset[1],\dots,\dset[n]$. This can be divided by $n$ to estimate the first marginal. 
    \item In \emph{phase two}, $\alg$ `clears the slate' by sending $-i$ to $\mech$ for each $\dset[i]$ with a $1$ in the first attribute. 
\end{itemize}
  It repeats this for each attribute, collecting the answers from $\mech,$ and then outputs its estimates for the marginals. In actuality, in both phases  of estimating the $j^{\text{th}}$ marginal, $\alg$ inputs $\perp$ for each $\dset[i]$ that has a $0$ in the $j^{\text{th}}$ attribute. This algorithm is $(O(\eps),O(\delta))$-DP for $\marginals{n,d}$ since changing one record $\dset[i]$ in the input to the algorithm $\cA$ will only change occurrences of a single element $i$ (to some other element $j$) in the input to the mechanism $\mech$. Additionally, note that this reduction works equally well in the likes model where items can only be inserted when absent and deleted when present, since the stream produced in the reduction has this structure (see \Thm{likes_lower_bound}).

\begin{algorithm}[ht!]
    \caption{Reduction $\alg$ from $\marginals{}{}$ to 
    $\countdistinct$}
    \label{alg:CD_item_reduction}
    \begin{algorithmic}[1]
        \Statex \textbf{Input:} Dataset $\dset = (\dset[1],\dots,\dset[n]) \in \zo^{n \times d}$ and black-box access to mechanism $\mech$ for $\countdistinct$ in turnstile streams
        \Statex \textbf{Output:}  Estimates of marginals $b = (b[1],\dots,b[d]) \in \R^d$
        \State Define the universe $\mathcal{U} = [n] \cup \{\perp\}$
        \State Initialize streams $z^{(1)}=\perp^{2n},\dots,z^{(d)}=\perp^{2n}$ and a vector $r$ of length $2nd$ 
        \ForAll{$(i,j) \in [n]\times [d]$ such that $\dset[i][j]=1$}
  
        \State Set $z^{(j)}[i] = +i$. 
        \Comment{phase one}
        \State Set $z^{(j)}[n+i] = -i$.
        \Comment{phase two}
        \EndFor
        \State Run $\mech$ on the stream $\dstream \gets z^{(1)} \circ z^{(2)} \circ \dots \circ z^{(d)}$ and record the answers as vector $r$ \label{step:construct-stream}  
        \ForAll{$j\in[d]$}
        \State $b[j] = r[(2j-1)n]/n$ \label{step:estimate-marginals}
        \EndFor
        \State Return estimates $(b[1],\dots, b[d])$ 
    \end{algorithmic}
\end{algorithm}


\begin{definition}[Accuracy of an algorithm for marginals]\label{def:batchmodel-accuracy}
Let $\gamma \in [0,1]$ and $n,d\in\N$.
The error $\err_{\marginals{}}$ is defined as in \Sec{intro}. A batch algorithm 
$\alg$ 
is $\gamma$-accurate for $\marginals{n,d}$ if for all datasets 
$y \in \{0,1\}^{n\times d}$,
$$\Pr_{{\text{coins of }}\alg}\left[ \err_{\marginals{}}(\dset, \alg(\dset)) \leq \gamma\right] \geq 0.99 . $$ 

\end{definition}

\begin{lemma}\label{lem:CD_item_reduction}
Let $\alg$ be \Alg{CD_item_reduction}. For all  $\eps \in (0,1], \delta \geq 0, \alpha \in \R^{+}$ and $d,n,w,T \in \N$, where $T \geq 2dn$ and $w \geq 2d$,  if  mechanism $\mech$  is  $(\eps, \delta)$-item-level-DP and $\alpha$-accurate for $\countdistinct$ for streams of length $T$ with maximum flippancy at most $w$ in the \CR{}, then batch algorithm~$\alg$ is $(2\eps, 4\delta)$-DP and $\frac{\alpha}{n}$-accurate for $\marginals{n,d}$. 
\end{lemma}

\begin{proof}
We start by reasoning about privacy. Fix neighboring datasets $\dset$ and $\dset'$ that are inputs to batch algorithm~$\alg$~(\Alg{CD_item_reduction}). (Datasets $y$ and $y'$ differ in one row.)
Let $\dstream$ and $\dstream'$ be the streams constructed in Step~\ref{step:construct-stream} of $\alg$ when it is run on $\dset$ and $\dset'$, respectively. 
We show that $\dstream$ and $\dstream'$ are $2$-item-neighbors. Suppose (w.l.o.g.) that $y$ and $y'$ differ in the first row. Then $\dstream$ and $\dstream'$ can only differ in the entries $z^{(j)}[1]$, where $j \in [d]$. The stream $\dstream'$ can be obtained from $\dstream$ in a sequence of two item-level changes: first, replace all entries $z^{(j)}[1]$ for $j \in [d]$ with $\perp$, then replace a subset of these entries with the appropriate insertions and deletions, following phase one and two of \Alg{CD_item_reduction}, to form $x'$. Thus $\dstream$ and $\dstream'$ are $2$-item-neighbors.
Since $\mech$ is $(\eps, \delta)$-item-level-DP, and $\alg$ only post-processes the outputs received from $\mech$, closure under post-processing (\Lem{postprocess}) and group privacy (\Lem{group_privacy})
imply that $\alg$ is $(2\eps, 4\delta)$-DP (where we use the fact that $\frac{e^{2\eps}-1}{e^{\eps}-1} \leq 4$ for all $\eps \in (0,1]$).

Now we reason about accuracy. Let $\dstream = (\dstream[1],\dots,\dstream[2dn])$ be the input stream provided to $\mech$ when $\alg$ is run on dataset $\dset.$ Recall that $\countdistinct(\dstream)[t]$ is the number of distinct elements in stream $\dstream$ at time $t$.
By construction, for all $j \in [d]$, the $j^{\text{th}}$ marginal $q_j(\dset)$ of the dataset $\dset$ is related to $\countdistinct(\dstream)[(2j-1)n]$ as follows

\begin{equation}\label{eq:countdist_obs}
    q^{(j)}(\dset) = \frac{1}{n}\sum_{i\in[n]} \dset[i][j]
    =  \frac{1}{n} \cdot \countdistinct(\dstream)[(2j-1)n].
\end{equation}

Notice that: (1) The coins of $\alg$ are the same as the coins of $\mech$ (since the transformation from $\mech$ to $\alg$ is deterministic). (2) The marginals are computed in Step~\ref{step:estimate-marginals} of \Alg{CD_item_reduction} 
using the relationship described by~Equation~(\ref{eq:countdist_obs}). (3) The maximum flippancy of the stream constructed in \Alg{CD_item_reduction} is at most $2d$, since each item $i\in\univ$ is inserted and deleted at most once in each $z^{(j)}$ for $j\in[d]$. We obtain that $\cA$ inherits its probability of success from $\mech$: 
\begin{align*}
 \Pr_{{\text{coins of }}\alg}\left[ \err_{\marginals{n,d}}(\dset, \alg(\dset)) \leq \frac \alpha n\right]
= &\Pr_{\text{coins of }\alg}\left[\max_{j\in[d]}\abs{q_j(\dset)-b[j]} \leq \frac{\alpha}{n} \right] \\
= &\Pr_{\text{coins of }\mech}\left[ \max_{t\in\{n,\dots,(2j-1)n\}} \abs{ \countdistinct(\dstream)[t] - r[t]} \leq \alpha \right] \\
\geq &\Pr_{\text{coins of }\mech}\left[ \max_{t\in[T]} \abs{ \countdistinct(\dstream)[t] - r[t]} \leq \alpha \right] 
\\
= & \Pr_{\text{coins of }\mech}\left[ \err_{\countdistinct}(\dstream, r)  \leq \alpha\right]
\geq 0.99,
\end{align*}
where we used that $\mech$ is $\alpha$-accurate for $\countdistinct$ for streams of length $T$ with maximum flippancy at most $w \leq 2d$. Thus, \Alg{CD_item_reduction} is $(\frac \alpha n)$-accurate for $\marginals{n,d}$, completing the proof of \Lem{CD_item_reduction}.
\end{proof}

\subsection{From the reduction to the accuracy lower bound}\label{sec:complete_lb_item}

In this section, we use \Lem{CD_item_reduction} (the reduction from  $\marginals{}$ to $\countdistinct$) together with previously established lower bounds for $\marginals{}$ to complete the proof of \Thm{IPlowerbound}. 
The lower bounds on the accuracy of private algorithms for $\marginals{}$ are stated in Items~1 and 2 of \Lem{oneway-marginals} for
approximate differential privacy and pure differential privacy, respectively. 
Item~2 in \Lem{oneway-marginals} is a slight modification of the lower bound from \citet{HardtT10} and follows from a simple packing argument.
\begin{lemma}[Lower bounds for $\marginals{}$ \cite{BunUV18, HardtT10}]
\label{lem:oneway-marginals}
For all \sstext{$\eps \in (0,2]$} \ssnote{Check range of $\eps$}, $\delta \in [0,1]$, $\gamma \in (0,1)$, $d,n \in \N$,  and algorithms that are $(\eps, \delta)$-differentially private and $\gamma$-accurate for $\marginals{n,d}$, the following statements hold.\\ 
\hspace*{3mm} {\bf 1} {\em \cite{BunUV18}.} If $\delta > 0$ and $\delta = o(1/n)$, then $n= \Omega\left( \frac{\sqrt{d}}{\gamma\eps  \log d} \right)$.\\
\hspace*{3mm}  {\bf 2} {\em \cite{HardtT10}.} If $\delta = 0$, then $n = \Omega\left( \frac{d}{\gamma \epsilon} \right)$.
\end{lemma}

To prove \Thm{LB_CD_item}, we show that the lower bound holds for $\eps=1$, and use \Thm{epsred} to extend it to all $\eps < 1$. The approximate-DP lower bound (on the error term $\alpha$) in \Thm{LB_CD_item} is the minimum of two terms. To prove this bound, we need to establish that, for every possible range of parameters, at least one term serves as a lower bound for $\alpha$. 

\begin{lemma}\label{lem:LB_CD_item_eps1}
For all $\delta  \in (0,1]$, sufficiently large $w,T \in \N$ such that $w \leq T$, and all $(1, \delta)$-item-level-DP
mechanisms that are $\alpha$-accurate for $\mathsf{CountDistinct}$ on turnstile streams of length $T$ with maximum flippancy at most $w$,\\
\hspace*{3mm} {\bf 1} If $\delta > 0$ and $\delta = o(1/T)$, then $\alpha = \Omega\Big(\mathsf{min} \left(\frac{\sqrt{w}}{\log w}, \frac{T^{1/3}}{\log T}\right)\Big)$.\\
\hspace*{3mm}  {\bf 2} If $\delta = 0$, then $\alpha = \Omega\Big(\mathsf{min} \left( w, \sqrt{T} \right)\Big)$.
\end{lemma}

\begin{proof} Let $\alg$ be the algorithm for $\marginals{n,d}$ with black-box access to an $\alpha$-accurate mechanism $\mech$ for $\countdistinct$, 
as defined in \Alg{CD_item_reduction}. 
If $T \geq 2dn$ and $w \geq 2d$, then by \Lem{CD_item_reduction}, algorithm $\alg$ is $(2,4\delta)$-differentially private and $(\frac{\alpha}{n})$-accurate  for $\marginals{n,d}$. 
We use \Lem{oneway-marginals} to lower bound $\alpha$.

\paragraph{Case~1 (Approximate DP, $\delta >0$, $\delta = o(1/n)$) :}
Suppose $w \leq T^{2/3}$. Pick number of dimension{s} $d = w/2$ and number of records $n=\frac{T}{w}$ (so that $T = 2dn$).
If $\frac{\alpha}{n}<1$, then by
Item~1 of \Lem{oneway-marginals}, $n = \Omega \left( \frac{n\sqrt{d}}{\alpha \log d} \right)$ which means that $\alpha = \Omega \left( \frac{\sqrt{d}}{ \log d} \right) = \Omega \left( \frac{\sqrt{w}}{\log w} \right)$. Otherwise, $\alpha \geq n \implies \alpha \geq \frac{T}{w} \geq T^{1/3} \geq \frac{T^{1/3}}{\log w} \geq \frac{\sqrt{w}}{\log w}$.

Now suppose $w=T^{2/3}$. The above argument
gives a lower bound of $\Omega \left( \frac{\sqrt{T^{2/3}}}{\log T^{2/3}} \right)$ on the accuracy of $\mech$. This lower bound applies to all $w > T^{2/3}$, since a mechanism that is $\alpha$-accurate for streams with maximum flippancy at most $w > w'$ is also $\alpha$-accurate for streams with maximum flippancy at most $w'$.

\paragraph{Case~2 (Pure DP, $\delta=0$):} The proof for $\delta = 0$ 
is similar,
except that we consider the cases $w \leq \sqrt{T}$ and $w > \sqrt{T}$ and use Item~2 from \Lem{oneway-marginals} instead of Item~1: 
Suppose $w \leq \sqrt{T}$.
Pick a dimension $d = w/2$, and number of entries $n=\frac{T}{w}$. 
If $\frac{\alpha}{n}<1$, then by \Lem{CD_item_reduction} and Item~1 of \Lem{oneway-marginals}, $n = \Omega \left( \frac{n\cdot d}{\alpha\cdot\epsilon} \right)$ which means that $\alpha = \Omega \left( \frac{d}{\epsilon} \right) = \Omega(w)$. Otherwise, if $\alpha \geq n$, then $\alpha \geq \frac{T}{w} \geq \sqrt{T} \geq w$.

Now, suppose $w \geq \sqrt{T}$. Since $\mech$ is also $\alpha$-accurate for streams of length $T$ with maximum flippancy $w' = \sqrt{T}$, the bound for $w \leq \sqrt{T}$ still applies: That is $\alpha = \Omega(w') \implies \alpha = \Omega(\sqrt{T})$.

This concludes the proof of \Lem{LB_CD_item_eps1}.
\end{proof}

Finally, we extend the lower bounds for $\eps=1$ from \Lem{LB_CD_item_eps1} to the general case of $\eps < 1$ using \Thm{epsred}.

\begin{proof}[Proof of \Thm{LB_CD_item}]
Suppose $\eps < \frac{2}{T}$. For these values of $\eps$, we prove an error lower bound of $\Omega(T)$, via a group privacy argument that is exactly the same as in the event-level lower bound (we direct the reader to the proof of \Thm{insertion+deletion+eps} for more details).

Now suppose $\eps \geq \frac{2}{T}$. 
For $\delta > 0$, 
\Lem{LB_CD_item_eps1} provides a lower bound of $\alpha' = \Omega\Big(\mathsf{min} \left(\frac{\sqrt{w}}{\log w}, \frac{T^{1/3}}{\log T}\right)\Big)$ on accuracy for 
($\eps' = 1$,$\delta' = o(1/T)$)-item-level-DP, $\alpha'$-accurate mechanisms for $\countdistinct$ on turnstile streams of length $T$ with maximum flippancy at most $w$.
By invoking \Thm{epsred}, we can extend this to the following lower bound for $(\eps, \delta)$-DP mechanisms where $\delta = \frac{\delta'\eps}{2} = o(\frac{\eps}{T})$:
    \begin{align*}
    \alpha =\frac{1}{\eps}\Omega\left(\mathsf{min}\left(\frac{\sqrt{w}}{\log w}, \frac{(\eps T)^{1/3}}{\log \eps T}\right)\right) =
    \Omega\left(\mathsf{min}\left(\frac{\sqrt{w}}{\eps\log w}, \frac{ T^{1/3}}{\eps^{2/3}\log \eps T}\right)\right) .
    \end{align*}
In different parameter regimes, we get lower bounds $\Omega(T)$ and 
$\Omega\left(\mathsf{min}\left(\frac{\sqrt{w}}{\eps\log w}, \frac{ T^{1/3}}{\eps^{2/3}\log \eps T}\right)\right)$.
Overall, we get a lower bound of 
$\Omega\left(\mathsf{min}\left(\frac{\sqrt{w}}{\eps\log w}, \frac{ T^{1/3}}{\eps^{2/3}\log \eps T},T\right)\right)$.

Finally, consider $\delta = 0$. \Lem{LB_CD_item_eps1} provides a lower bound of $\alpha' = \Omega\Big(\mathsf{min} \left( w, \sqrt{T} \right)\Big) $ on accuracy for 
($\eps' = 1$, $\delta=0$)-item-level-DP, $\alpha'$-accurate mechanisms for $\countdistinct$ on turnstile streams of length $T$ with maximum flippancy at most $w$. By invoking \Thm{epsred} with $\delta=0$, we can extend this to the following lower bound for $(\eps, 0)$-DP mechanisms:
    \begin{align*}
    \alpha = \Omega\Big(\mathsf{min} \Big(\frac{w}{\eps}, \frac{\sqrt{T}}{\eps} \Big)\Big). 
    \end{align*}
    Combining this bound with the $\Omega(T)$ lower bound for when $\eps < \frac{2}{T}$ gives the desired result.
\end{proof}

\newcommand{\acksection}{\section*{Acknowledgments}}

\acksection{
We thank Teresa Steiner and an anonymous reviewer for useful suggestions on the initial version of this paper. S.S. was supported by NSF award CNS-2046425 and Cooperative Agreement CB20ADR0160001 with the Census Bureau. A.S. and P.J. were supported in part by NSF awards CCF-1763786 and CNS-2120667 as well as Faculty Awards from Google and Apple.
}

\bibliographystyle{plainnat}
\bibliography{biblio}

\newpage
\appendix

\section{General lower bound reduction for small $\eps$}\label{sec:smalleps}

We describe an adaptation of a folklore reduction to our problem of interest that allows us to extend a lower bound for $\eps=1$ to any $\eps < 1$. The theorem is stated for item-level differential privacy, but applies to event-level differential privacy as well. In addition, this reduction also works for variants of the model previously discussed: offline mechanisms, the likes model, and the strict turnstile model. 

\begin{theorem}\label{thm:epsred}
    Let $\alpha: \R \to \R$ be a nondecreasing function. Let $T \in \mathbb{N}$, and $\eps, \delta \in [0,1]$ such that $\eps \geq \frac 1 T$. If for all $T\in \N$, every mechanism for $\countdistinct{}$ that is $(1,\delta)$-item-level-DP for streams of length $T$ has error parameter 
    at least $\alpha(T)$, then every mechanism for $\countdistinct{}$ that is $(\eps, \frac{\delta \eps}{2})$-item-level-DP for streams of length $T'$ has error parameter 
    at least $\frac{\alpha(\eps T')}{2\eps}$.
\end{theorem}

\begin{proof} 
    Fix $\eps\in(0,1]$ and  $\delta \in [0,1)$. Let $\ell = \lfloor1/\eps\rfloor$ and $T' \in \mathbb{N}$, where $T'$ is divisible by $\ell$.
    %
    We prove the contrapositive. Namely,  
    let $\cM'$ be a mechanism for $\countdistinct$ that is $\frac{\alpha(\eps T')}{2\eps}$-accurate and $(\eps, \frac{\delta \eps}{2})$-item-level-DP for streams of length $T'$. We use $\cM'$ to construct a mechanism $\cM$ for $\countdistinct$ that is $\alpha(T)$-accurate and $(1, \delta)$-item-level-DP for streams of length $T = 
    {T'/\ell}$.

    Given a stream $\dstream \in \univ_\pm^T$ of length $T= 
    {T'/\ell}$, create a new stream $\dstream'$ of length $T'$ with insertions and deletions from a larger universe $\univ' =  \univ \times [\ell]$ (so every item in $\univ$ corresponds to $\ell$ distinct items in $\univ'$), as follows:
    Initialize $\dstream'$ to be empty. 
    For $t \in [T]$, if $\dstream[t] = +i$, 
    append $+(i,1), \ldots, +(i,\ell)$ to $\dstream'$. Similarly, if $\dstream[t] = -i$, append $-(i,1), \ldots, -(i,\ell)$.
    If $\dstream[t] = \perp$, append $\ell$ stream entries 
    $\perp$ to  $\dstream'$. Finally, define the output of $\cM$ as follows: run 
    $\cM'$ on $\dstream'$ and,
    for each time step $t\in T$, output $\cM(\dstream)[t] = 
    \frac 1 \ell \cdot \cM'(\dstream')[t\cdot \ell]$.

    Changing a single item of $\dstream$ changes
    $\ell$ items of $\dstream'$. Therefore, by the privacy guarantee of $\cM'$, by group privacy (see \Lem{group_privacy}) and since $e^\eps-1\geq \eps$ for all $\eps$, we get that $\cM$  is 
   $(\eps^*,\delta^*)$
    -item-level-DP for streams of length $T = 
    {T'/\ell}$, where $\eps^*=\ell\eps\leq 1$ and  $\delta^*=\frac{\delta \eps}{2}\cdot \frac{e^{\ell\eps}-1}{e^\eps -1}\leq \frac{\delta \eps} {2}\cdot\frac 2 \eps =\delta$. Thus, $\cM$ is $(1,\delta)$-item-level-DP for streams of length $T$.
    
    Finally, we derive the accuracy guarantee of $\cM$
 from the accuracy guarantee of $\cM'$. By the construction of $\dstream',$ $$\countdistinct(\dstream')[t\cdot\ell] = \ell\cdot \countdistinct(\dstream)[t].$$ 
    By the accuracy of $\cM'$, on every input stream $\dstream'$ of length $T'$, with probability at least $0.99$, the maximum error of $\cM'$  is at most $\frac{\alpha(\eps T'
    )}{2\eps}$. 
    Therefore, with probability at least $0.99$, the maximum error of $\cM$ is at most 
    $ \frac 1 \ell \cdot \frac{\alpha(\eps T')}{2\eps} \leq \frac 1 {2\eps \ell} \cdot \alpha(\eps\cdot \ell\cdot T) \leq \alpha(T)$, as desired. This concludes the proof of \Thm{epsred}. 
    %
    %
\end{proof}

\section{Lower bounds for variants of the model}\label{sec:other_lower_bounds}

In this section, we prove lower bounds for $\countdistinct$ in more restricted turnstile streaming models. In \Thm{likes_lower_bound}, we prove that our item-level lower bound applies to the more restricted likes model, whereas in \Thm{strict_lower_bound}, we prove that both our event-level and item-level lower bounds apply to the strict turnstile model. 
In addition, our lower bounds apply even when considering only \emph{offline} mechanisms for $\countdistinct$, defined next.

\begin{definition}[Offline mechanisms for $\countdistinct$]\label{def:offline}
    A mechanism for $\countdistinct$ is an \emph{offline} mechanism if it first receives the entire input stream $x$ before outputting the vector of $\countdistinct$ estimates.
\end{definition}

\begin{definition}[Likes model]\label{def:likes}
    In the likes model, an element can be inserted into the stream only when it is absent and deleted from the stream $x$ only when it is present. Mathematically, for a stream $x$ of length $T$, for all $t\in[T]$, if $x[t] = +u$ for some element $u \in \uni$, then $f_u(x)[t-1] = 0$, and if $x[t] = -u$, then $f_u(x)[t-1] = 1$, where we define  $f_u(x)[0] = 0$  for convenience. 
\end{definition}

\begin{theorem}[Item-level lower bound in the likes model]\label{thm:likes_lower_bound}
     The item-level lower bound in \Thm{LB_CD_item} holds for all offline mechanisms for $\countdistinct$ in the likes model.
\end{theorem}

\begin{proof}

    The proof proceeds identically to the proof of \Thm{LB_CD_item}, with the following observations. 
    First, note that our reduction in \Alg{CD_item_reduction} uses the answers returned by the mechanism $\cM$ for $\countdistinct$ only \emph{after} $\cM$ has produced the entire stream of answers (it does not require online access to the outputs of $\cM$). Therefore, one can restate the reduction so that it only requires black-box access to an offline mechanism $\mathcal{M}$ for $\countdistinct$, and the proof of \Thm{LB_CD_item} would still hold.
    
     Next, observe that the stream $x$ constructed in the reduction in \Alg{CD_item_reduction} belongs to the likes model. This holds because, for each element, the stream entries belonging to that element alternate between insertions and deletions, with the first entry always being an insertion. Finally, the lower bound reduction from $\eps=1$ to $\eps < 1$ in \Thm{epsred} works both for offline mechanisms and streams in the likes model. 
\end{proof}

\begin{definition}[Strict turnstile model]
    In the \emph{strict turnstile model}, an element can be deleted only when it is present in the stream.  Mathematically, for a stream $x$ of length $T$, for all $t\in[T]$, if $x[t] = -u$ for some time step $t$ and element $u \in \uni$, then $f_u(x)[t-1] = 1$, where we define $f_u(x)[0] = 0$ for convenience. 
\end{definition}

\begin{theorem}[Lower bounds for the strict turnstile model]\label{thm:strict_lower_bound}
     The event-level lower bound of \Thm{insertion+deletion+eps} and the item-level lower bound of \Thm{LB_CD_item} hold for all offline mechanisms for $\countdistinct$ in the strict turnstile model.
\end{theorem}

\begin{proof}
    Since the likes model is a special case
    of the strict turnstile model, by \Thm{likes_lower_bound}, the item-level lower bound applies to the strict turnstile model as well. We now focus on the event-level lower bound. Again, the proof proceeds identically to the proof of \Thm{insertion+deletion+eps}, with the following observations. 
    
    First, note that our reduction in \Alg{innerprodred}, similarly to the reduction in \Alg{CD_item_reduction}, uses the outputs of the mechanism $\cM$ for $\countdistinct$ only \emph{after} the mechanism has produced the entire stream of answers. Therefore, one can restate the reduction so that it only requires black-box access to an offline mechanism $\mathcal{M}$ for $\countdistinct$.
    
    Next, the stream $x$ constructed in Step~\ref{step:construct_x_event} of the reduction in \Alg{innerprodred} belongs to the strict turnstile model, since each element deletion is preceded by an insertion of the same element. Finally, the lower bound reduction from $\eps=1$ to $\eps < 1$ in \Thm{epsred} works both for offline mechanisms and streams in the strict turnstile model. 
\end{proof}

\section{A concentration inequality for Gaussian random variables}\label{sec:concentration}

Recall that $\mathcal{N}(\mu, \sigma^2)$ denotes the Gaussian distribution with mean $\mu$ and standard deviation $\sigma$.

\begin{lemma}\label{lem:gauss_conc}
For all random variables $R \sim \cN(0,\sigma^2)$,
\begin{equation*}
    \Pr[|R| > \lambda] \leq 2e^{-\frac{\lambda^2}{2 \sigma^2}}.
\end{equation*}
\end{lemma}
\begin{lemma}\label{lem:gauss_max}
Consider $m$ random variables $R_1,\dots,R_m \sim \cN(0,\sigma^2)$. Then
\begin{equation*}
    \Pr[\max_{j \in [m]} |R_j| > \lambda] \leq 2 m e^{-\frac{\lambda^2}{2 \sigma^2}}.
\end{equation*}
\end{lemma}
\begin{proof}
By a union bound and \Lem{gauss_conc},
\begin{align*}
    \Pr[\max_{j \in [m]} |R_i| > \lambda] 
      &= \Pr[\exists i \in [m] \text{ such that } |R_i| > \lambda] \\
     &\leq \sum_{j=1}^m \Pr[ |R_i| > \lambda] \leq
    \sum_{j=1}^m  2e^{-\frac{\lambda^2}{2 \sigma^2}}  =  2m e^{-\frac{\lambda^2}{2 \sigma^2}}.
    \hspace{3cm} \hfill
    \qedhere
\end{align*}
\end{proof}

\end{document}